\documentclass[11pt]{article}

\usepackage{bbm}
\usepackage{framed} 
\usepackage{url} 
\usepackage{complexity}
\usepackage{booktabs}
\usepackage{amsmath,amssymb}
\usepackage{amsthm}
\usepackage{thmtools} 
\usepackage{thm-restate}
\usepackage{nicefrac}
\usepackage{microtype}
\usepackage{calc}
\usepackage{enumerate}
\usepackage{enumitem}
\usepackage[usenames,dvipsnames]{xcolor}
\usepackage[colorlinks,citecolor=blue,linkcolor=BrickRed]{hyperref}
\usepackage{textgreek}
\usepackage{setspace}
\usepackage{parskip}

\usepackage{fullpage}
\allowdisplaybreaks

\usepackage{graphicx}
\usepackage[font=footnotesize,labelfont=bf]{subcaption}
\usepackage[font=footnotesize,labelfont=bf]{caption}
\usepackage[nobreak=true,skipabove=7pt]{mdframed}
\usepackage{appendix}
\usepackage[colorinlistoftodos]{todonotes}

\usepackage{algorithm}
\usepackage{algcompatible}
\usepackage[noend]{algpseudocode}

\usepackage{xr}
\usepackage{array}
\usepackage{xspace}

\usepackage[capitalise]{cleveref}
\crefrangelabelformat{enumi}{#3#1#4--#5#2#6}
\usepackage{chngcntr}
\usepackage{mathtools}

\crefname{claim}{Claim}{Claims}

\usepackage{nameref}

\usepackage{tikz}
\usetikzlibrary{patterns}

\DeclareMathOperator*{\Ber}{Ber}
\DeclareMathOperator*{\argmin}{argmin}

\DeclareMathOperator*{\expectation}{\mathbb{E}}
\let\poly\relax
\DeclareMathOperator*{\poly}{poly}

\DeclareMathOperator*{\probability}{\mathbb{P}}

\newcommand\Z{\mathbb{Z}}

\newcommand\eps{\epsilon}

\DeclarePairedDelimiterX{\probarg}[1]{(}{)}{%
	\ifnum\currentgrouptype=16 \else\begingroup\fi
	\activatebar#1
	\ifnum\currentgrouptype=16 \else\endgroup\fi
}

\newcommand{\probover}[1]{\probability_{#1}\probarg}

\newcommand{\expect}{\expectation\expectarg}
\DeclarePairedDelimiterX{\expectarg}[1]{[}{]}{%
	\ifnum\currentgrouptype=16 \else\begingroup\fi
	\activatebar#1
	\ifnum\currentgrouptype=16 \else\endgroup\fi
}

\newcommand{\expectover}[1]{\expectation_{#1}\expectarg}

\DeclarePairedDelimiterX{\wklarg}[1]{(}{)}{%
	\ifnum\currentgrouptype=16 \else\begingroup\fi
	\activatebars#1
	\ifnum\currentgrouptype=16 \else\endgroup\fi
}

\newcommand{\innermid}{\nonscript\;\delimsize\vert\nonscript\;}
\newcommand{\activatebar}{%
	\begingroup\lccode`\~=`\|
	\lowercase{\endgroup\let~}\innermid 
	\mathcode`|=\string"8000
}

\newcommand{\innermids}{\nonscript\;\delimsize\vert\delimsize\vert\nonscript\;}
\newcommand{\activatebars}{%
	\begingroup\lccode`\~=`\|
	\lowercase{\endgroup\let~}\innermids 
	\mathcode`|=\string"8000
}

\newcommand\opt{\textsc{Opt}\xspace}
\newcommand\copt{c\left(\textsc{Opt}\right)\xspace}
\newcommand\opto{\textsc{Opt}_\textsc{Online}\xspace}
\newcommand\coptof[1]{c\left(\textsc{Opt}\left(#1\right)\right)\xspace}
\newcommand\lpopt{\textsc{LP}_{\textsc{Opt}}}
\newcommand\alg{\textsc{Alg}\xspace}

\counterwithin{equation}{section}

\newcommand{\ForEach}{\textbf{for each}\xspace}

\usepackage{eqparbox}

\newcommand\lastt{{t-1}}
\newcommand\thist{{t}}

\newcommand\KL[2]{\textsc{KL}\left(#1 \mid \mid #2\right)}
\newcommand\wKL[2]{\textsc{KL}_c\left(#1 \mid \mid #2\right)}
\newcommand\supp[1]{\text{support}\left( #1 \right)}

\newcommand\loc{\textsc{LearnOrCover}\xspace}

\newcommand\setcov{\textsc{SetCover}\xspace}

\newcommand\backup{\textsc{Backup}\xspace}
\newcommand\aug[2]{\textsc{Aug}\left(#1 \ \middle| \ #2\right)\xspace}
\newcommand\aip{\textsc{AIP}\xspace}
\newcommand\cip{\textsc{CIP}\xspace}
\newcommand\nscws{$(\alpha,\delta)\textsc{-NoisySampleSetCover}$\xspace}
\newcommand\facloc{\textsc{FacilityLocation}\xspace}
\newcommand\nmfl{\textsc{NonMetricFacilityLocation}\xspace}
\newcommand\mfl{\textsc{MetricFacilityLocation}\xspace}
\newcommand\setmcov{\textsc{SetMultiCover}\xspace}
\theoremstyle{plain}

\newtheorem{theorem}{Theorem}[section]
\newtheorem{definition}[theorem]{Definition}

\newtheorem{fact}[theorem]{Fact}
\newtheorem{observation}[theorem]{Observation}

\newtheorem{corollary}[theorem]{Corollary}

\newtheorem{invariant}{Invariant}

\newlength{\continueindent}
\usepackage{etoolbox}
\makeatletter
\newcommand*{\ALG@customparshape}{\parshape 2 \leftmargin \linewidth \dimexpr\ALG@tlm+\continueindent\relax \dimexpr\linewidth+\leftmargin-\ALG@tlm-\continueindent\relax}
\apptocmd{\ALG@beginblock}{\ALG@customparshape}{}{\errmessage{failed to patch}}
\makeatother

\makeatletter
\def\thm@space@setup{%
	\thm@preskip=\parskip \thm@postskip=0pt
}
\makeatother

\usepackage{etoolbox}
\usepackage{tikz}
\usetikzlibrary{tikzmark}
\usetikzlibrary{calc}

\errorcontextlines\maxdimen

\newcommand{\ALGtikzmarkcolor}{black}
\newcommand{\ALGtikzmarkextraindent}{4pt}
\newcommand{\ALGtikzmarkverticaloffsetstart}{-.5ex}
\newcommand{\ALGtikzmarkverticaloffsetend}{-.5ex}
\makeatletter
\newcounter{ALG@tikzmark@tempcnta}

\newcommand\ALG@tikzmark@start{%
	\global\let\ALG@tikzmark@last\ALG@tikzmark@starttext%
	\expandafter\edef\csname ALG@tikzmark@\theALG@nested\endcsname{\theALG@tikzmark@tempcnta}%
	\tikzmark{ALG@tikzmark@start@\csname ALG@tikzmark@\theALG@nested\endcsname}%
	\addtocounter{ALG@tikzmark@tempcnta}{1}%
}

\def\ALG@tikzmark@starttext{start}
\newcommand\ALG@tikzmark@end{%
	\ifx\ALG@tikzmark@last\ALG@tikzmark@starttext
	\else
	\tikzmark{ALG@tikzmark@end@\csname ALG@tikzmark@\theALG@nested\endcsname}%
	\tikz[overlay,remember picture] \draw[\ALGtikzmarkcolor] let \p{S}=($(pic cs:ALG@tikzmark@start@\csname ALG@tikzmark@\theALG@nested\endcsname)+(\ALGtikzmarkextraindent,\ALGtikzmarkverticaloffsetstart)$), \p{E}=($(pic cs:ALG@tikzmark@end@\csname ALG@tikzmark@\theALG@nested\endcsname)+(\ALGtikzmarkextraindent,\ALGtikzmarkverticaloffsetend)$) in (\x{S},\y{S})--(\x{S},\y{E});%
	\fi
	\gdef\ALG@tikzmark@last{end}%
}

\apptocmd{\ALG@beginblock}{\ALG@tikzmark@start}{}{\errmessage{failed to patch}}
\pretocmd{\ALG@endblock}{\ALG@tikzmark@end}{}{\errmessage{failed to patch}}
\makeatother

\algblock[with]{With}{EndWith}
\algblockdefx[With]{With}{EndWith}%
[1]{\textbf{with} #1 \textbf{do}}%
{}

\makeatletter
\ifthenelse{\equal{\ALG@noend}{t}}%
{\algtext*{EndWith}}
{}%
\makeatother

\title{Set Covering with Our Eyes Wide Shut}

 \author{Anupam Gupta\thanks{Computer Science Department, Carnegie Mellon
 		University, Pittsburgh, PA 15213. Email:
 		\texttt{anupamg@cs.cmu.edu}. Research supported in part by NSF awards CCF-1955785, CCF-2006953, and CCF-2224718.}
 	\and
 	Gregory Kehne\thanks{School of Engineering and Applied Sciences, Harvard
 		University, Boston, MA 02138. Email:
 		\texttt{gkehne@g.harvard.edu}. Research supported in part by the Siebel Scholars program.}
 	\and
 	Roie Levin\thanks{Department of Statistics and Operations Research, School of Mathematical Sciences, Tel Aviv University, Tel Aviv. Email:
 		\texttt{roiel@tauex.tau.ac.il}. Supported in part by a Fulbright Israel Postdoctoral Fellowship, Israel Science Foundation grant 2233/19, and United States - Israel Binational Science Foundation grant 2018352.}
 }
\date{}

\begin{document}

	\maketitle
	
    	\begin{abstract}
          In the \emph{stochastic set cover problem} (Grandoni et al.,
          FOCS '08), we are given a collection $\mathcal{S}$ of $m$
          sets over a universe $\mathcal{U}$ of size $N$, and a
          distribution $D$ over elements of $\mathcal{U}$. The algorithm draws $n$ elements one-by-one from $D$ and must buy a set to cover each element on arrival; the goal is to minimize the total cost of sets bought during this process. A \emph{universal} algorithm a-priori maps each element
          $u \in \mathcal{U}$ to a set $S(u)$ such that if $U \subseteq \mathcal{U}$ is
          formed by drawing $n$ times from distribution $D$, then the algorithm commits to outputting $S(U)$. Grandoni et al. gave an $O(\log mN)$-competitive
          universal algorithm for this stochastic set cover problem.

          We improve unilaterally upon this result by
          giving a simple, polynomial time $O(\log mn)$-competitive
          universal algorithm for the more general \emph{prophet}
          version, in which $U$ is formed by drawing from $n$
           different distributions $D_1, \ldots,
          D_n$. Furthermore, we show that we do not need full
          foreknowledge of the distributions: in fact, a \emph{single
            sample} from each distribution suffices. We show similar
          results for the \emph{2-stage prophet} setting and for the
          \emph{online-with-a-sample} setting.

          We obtain our results via a generic reduction from the
          single-sample prophet setting to the random-order setting;
          this reduction holds for a broad class of minimization
          problems that includes all covering problems. We take
          advantage of this framework by giving random-order
          algorithms for non-metric facility location and set
          multicover; using our framework, these automatically
          translate to universal prophet algorithms. 
	\end{abstract}

    \newpage

\section{Introduction}

In the \setcov problem we are given a set system $(\mathcal{U},\mathcal{S})$, where $\mathcal{U}$ is a ground set of size $N$ and $\mathcal{S}$ is a collection of subsets with $|\mathcal{S}| = m$.  We are also given a subset $U \subseteq \mathcal{U}$ of size $n$. The goal is to select a minimum-size (or more generally, minimum-cost) subcollection $\mathcal{S}' \subseteq \mathcal{S}$ such that the union of the sets in $\mathcal{S}'$ is $U$. Many polynomial-time algorithms have been discovered for this problem that achieve an approximation ratio of $\ln n$ (see e.g. \cite{chvatal1979greedy,johnson1974approximation,lovasz1975ratio,williamson2011design}), and this is best possible unless $\P = \NP$ \cite{DBLP:journals/jacm/Feige98,dinur2014analytical}.

One may interpret a solution $\mathcal{S}'$ as a map $\mathfrak{S}: \mathcal{U} \rightarrow \mathcal{S}$ taking each element to a set that covers it (breaking ties arbitrarily). In this case $\mathfrak{S}(U) = \bigcup_{u \in U} \{\mathfrak{S}(u)\}$ is the solution $\mathcal{S'}$. In seminal work, Jia et al. defined the \emph{universal} variant of the set cover problem, in which the goal is to construct $\mathfrak{S}$ a priori and obliviously without seeing the actual value of $U\subseteq \mathcal{U}$ (hence it is constructed using only $\mathcal{U}$ and $\mathcal{S}$) \cite{DBLP:conf/stoc/JiaLNRS05}. 
One wants a map $\mathfrak{S}$ minimizing the worst case ratio  $\max_{U \subseteq \mathcal{U}} c(\mathfrak{S}(U)) / c(\opt(U))$ between the cost of $\mathfrak{S}(U)$ and the cost of the optimal set cover for $U$. 
A universal algorithm is said to be $\alpha$-competitive, or to achieve competitive ratio $\alpha$, if the value of this ratio is no more than $\alpha$. 
Jia et al. showed $\tilde \Theta(\sqrt{n})$ bounds for this problem \cite{DBLP:conf/stoc/JiaLNRS05}.

To overcome this polynomial barrier, Grandoni et al. \cite{DBLP:journals/siamcomp/GrandoniGLMSS13} studied the \emph{stochastic} variant of universal set cover, in which one additionally assumes that the elements of $U$ are drawn i.i.d. from a known distribution $D$. The aim is now to minimize the expected ratio $E_U[c(\mathfrak{S}(U))] / E_U[c(\opt(U))]$. With this assumption, they showed that it is possible to get an exponentially better $O(\log (mN))$-competitive algorithm, and that this is best possible up to $\log \log$ factors.

In this work, we improve, generalize and simplify the results of \cite{DBLP:journals/siamcomp/GrandoniGLMSS13}. 
First, we improve the competitive ratio to $O(\log (mn))$, which can be exponentially smaller when $n \ll N$. This essentially is best possible for polynomial-time algorithms, since there is an $\Omega(\log m/ \log \log m)$ lower bound when $n \ll m$ \cite{DBLP:journals/siamcomp/GrandoniGLMSS13}, and there is no polynomial-time algorithm with approximation $o(\log n)$ unless $\P = \NP$ \cite{DBLP:journals/jacm/Feige98,dinur2014analytical}. We also generalize to the \emph{prophet} setting, in which $U$ consists of draws from a sequence of non-identical distributions $D^1, \ldots, D^n$.
In fact, we show that we do not need full knowledge of these distributions, and even a single sample from each distribution suffices. We also show extensions to two other related models, the \emph{2-stage} prophet setting, and the online-with-a-sample setting, as well as to several problems which generalize \setcov and covering.
We now present a more formal overview of these results.

\subsection{Our Results}
Our main contribution is a reduction from the prophet setting to the random-order online setting.
In random-order \setcov, the elements $U$ are adversarially chosen and revealed one at a time.
The algorithm must choose a set to cover the element, and decisions are irrevocable. Since the \loc algorithm of \cite{DBLP:conf/focs/0001KL21} is an $O(\log (mn))$-competitive algorithm for random-order \setcov, we immediately obtain:

\begin{theorem}
\label{thm:main_prophet}
There is a polynomial-time $O(\log (mn))$-competitive universal algorithm for $1$-sample prophet \setcov.
\end{theorem}

Using similar techniques, we obtain theorems for the following two models as well. In the $2$-stage prophet model, the algorithm is allowed to purchase sets at a discount in a first stage before the game begins. In a second stage, $U$ is drawn at random as in the usual prophet setting, and any sets bought after seeing the realizations cost full price.

\begin{theorem}
\label{thm:main_2stage}
There is a polynomial-time universal algorithm for $2$-stage prophet \setcov that is $O(\log(mn))$-competitive with respect to the optimal online policy.
\end{theorem}

In the \emph{online-with-a-sample} setting, an adversary selects an unknown element set $U$ and reveals a uniformly random $\alpha$-fraction of it to the algorithm. After this point, the remaining elements are revealed one-by-one in adversarial order. The algorithm must buy sets to cover incoming elements immediately on arrival, and decisions are irrevocable.

\begin{theorem}
\label{thm:main_was}
For every $O < \alpha \leq 1$, there is a polynomial-time $O(\log(mn)/\alpha)$-competitive universal algorithm for online-with-a-sample \setcov.
\end{theorem}

In fact our reduction holds for a more general class of minimization problems which we call \emph{augmentable integer programs} ({\aip}s). These are problems for which augmenting any infeasible solution $z$ to a feasible solution gets cheaper as $z$ grows\footnote{Such problems were studied by \cite{DBLP:conf/soda/GargGLS08} but were not given a name. These are intuitively related to \emph{covering} problems; we reserve the term covering for problems with upward closed feasible regions. See \cref{sec:prelim} for details.}. Hence we reduce the task of designing prophet/$2$-stage prophet/online-with-a-sample algorithms for {\aip}s to the task of designing random-order algorithms. 
To illustrate our reductions, we additionally give $O(\log mn)$-competitive algorithms for \nmfl and \setmcov in random order. 
This marks partial progress in answering an open question of \cite{DBLP:conf/focs/0001KL21}, which asked if there is an $O(\log mn)$-competitive algorithm for covering integer programs with box constraints in random order.

\subsection{Techniques and Overview}

The proofs of \cref{thm:main_prophet,thm:main_2stage,thm:main_was} share a common template. The main idea is to reduce from setting $\mathcal{X}$ to random-order \setcov as follows: 
\begin{enumerate}
    \item Generate a mock instance $\widehat I \sim \mathcal{X}$.
    \item Simulate algorithm $\mathcal{A}$ for random-order \setcov on the mock input $\widehat I$ by shuffling the order artificially.
    \item Solve $\mathcal{X}$ on the real input $\mathcal{I} \sim \mathcal{X}$ by first buying the solution $z$ bought by $\mathcal{A}$, then covering any outstanding uncovered element $v^t$ with the cheapest set containing $v^t$.
\end{enumerate}
The idea is to charge the ``backup'' sets bought to cover any elements missed by $\mathcal{A}$ to the actual decisions made by $\mathcal{A}$, which we can bound using the performance guarantees on $\mathcal{A}$. This perspective allows us to give us proofs that are simple in hindsight; prior to our work it was not known how to obtain such results. We treat each model separately in \cref{sec:prophet,sec:2stage,sec:was}.

In \cref{sec:nmfl,sec:setmcov} we give random-order algorithms for \nmfl and \setmcov.
This demonstrates the generality of our reductions, and also illustrates the versatility of the ``Learn or Cover'' framework of \cite{DBLP:conf/focs/0001KL21} beyond pure covering problems.
These are the most technically involved sections of this work.

These results build upon the \loc framework of \cite{DBLP:conf/focs/0001KL21} designed for set cover. With every element that arrives uncovered, \loc (a) samples from a distribution over sets, and (b) learns from the fact that a random element was uncovered to update the distribution. \cite{DBLP:conf/focs/0001KL21} show the algorithm either makes progress \emph{learning} about the optimal distribution from which one should be sampling, or if it does not then it makes progress \emph{sampling} since the distribution is already sufficiently good. They used a two-part potential, where the parts measure progress learning and covering respectively.

Non-metric facility location is often treated as an extension of set cover, since there are standard reductions between the two (\cite[Section 3.1]{vygen2005approximation} and \cite{kolen1984covering}). 
However, these reductions do not hold in the random-order model. 
The first reduction from \nmfl to \setcov, which is folklore, requires an exponential blowup in the number of sets. This is prohibitive since one must in general lose a $\Omega(\log m)$ factor for random-order online set cover \cite{DBLP:conf/focs/0001KL21}. 
The second reduction introduces a new set and a new element for every facility-client pair; thus a client arriving in random order becomes a batch of new sets and elements.
However, online set cover in which elements arrive in randomly ordered \emph{batches} is in general harder than true uniform random order since \cite{DBLP:conf/focs/0001KL21} show a doubly logarithmic lower bound for  this problem. Therefore both reductions face obstacles in the random-order setting, and a new approach is needed.

The primary challenge in random-order \facloc is to account for \emph{connection costs}. This makes the task of learning a distribution over facilities complex, since the costs of satisfying arriving clients to change over time.
Our approach may be  viewed as running \loc on a set system that evolves dynamically over time: each facility is a set, each client is an element, and a client's element is contained in a facility's set if opening that facility significantly reduces that client's connection cost. We reuse the high level learn/cover idea, but we need to use a more intricate potential to measure progress learning.

Finally, our random-order \setmcov algorithm builds on the slightly more involved algorithm of \cite{DBLP:conf/focs/0001KL21} for random-order {\cip}s. This involves several technical challenges. For one, the two-part potential of \cite{DBLP:conf/focs/0001KL21} expects that if a variable's probability in the maintained distribution is high, then it will contribute towards covering unseen constraints in expectation. However multiplicity constraints prohibit the algorithm from sampling any variable more than once, even in this case. We show that this difficulty can nevertheless be circumvented by gradually ``forgetting'' coordinates that have already hit their caps; interestingly our multiplicative weights update rule does not depend on the marginal augmentation cost of the incoming constraint, as it does in \cite{DBLP:conf/focs/0001KL21}.

\subsection{Related Work}

    The term \emph{prophet inequality} is usually used in the context of online max finding: a gambler draws numbers one-by-one from a sequence of known distributions, and their task is to stop at the highest number. 
    Prophet inequality refers to the bound on the performance of such a player in terms of that of a clairvoyant ``prophet'' who can see the future. \cite{krengel1978semiamarts} showed a strategy for this game with expected reward at least $\nicefrac{1}{2}$ that of the prophet (see \cite{hill1992survey} for a further survey). 
    The \emph{secretary problem} \cite{ferguson1989solved} is a related max-finding game in which the gambler sees \emph{arbitrary} numbers in \emph{random} order, and once again aims to stop at the highest number.
    \cite{DBLP:conf/soda/AzarKW14} gave a $\nicefrac{1}{e}$-competitive strategy for the prophet problem (and extensions) via a reduction to (a subclass of algorithms for) the secretary problem, and this bound was later improved to $\nicefrac{1}{2}$ by \cite{DBLP:conf/innovations/RubinsteinWW20}.
    Our main result may be viewed as a minimization counterpart of the prophet-to-secretary reductions of \cite{DBLP:conf/soda/AzarKW14} for maximization problems.
    
    \emph{Free-order} prophet inequalities, in which the gambler can adaptively choose the order in which to open boxes, were studied by \cite{DBLP:conf/sigecom/LiuLPSS21,DBLP:conf/focs/PengT22,DBLP:journals/corr/abs-2211-04145}. Our main result implies that for covering problems, the constrained-order prophet problem is---up to a factor of two---no harder than its free-order counterpart.

    Previous work of \cite{dehghani2018greedy} claimed a reduction from the prophet set cover problem to universal algorithms for the stochastic (in other words i.i.d.) set cover problem. However the proof (which appears in Section 9.5 of \cite{ehsani2017online}) has an issue which we detail in \cref{sec:error}, and the claim has since been withdrawn \cite{saeedemails}.

    Motivated by settings where an algorithm has access to historical
    data, \cite{DBLP:conf/soda/KaplanNR22,DBLP:conf/soda/KaplanNR20}
    recently introduced the online-with-a-sample model in the context
    of max-finding (i.e. the secretary problem), and
    matching. \cite{AFGS22-neurips} study Steiner tree, facility
    location and load balancing in this model.

    There is considerable work on 2-stage (and more generally
    multi-stage) stochastic optimization from the perspective of
    approximations
    (e.g.,~\cite{DBLP:journals/jacm/ShmoysS06,DBLP:conf/approx/CharikarCP05,DBLP:conf/stoc/GuptaPRS04,DBLP:conf/approx/GuptaPRS05,DBLP:journals/mor/GuptaRS07}),
    see \cite{DBLP:conf/fsttcs/SwamyS06,birge2011introduction} for
    surveys. Our $2$-stage prophet model is a \emph{hybrid}
    stochastic-online model in which the second stage is a fully
    online game; as far as we know, this model has not been
    previously studied.

    Finally, using our reduction framework and the $3$-competitive random-order algorithm of \cite{kaplan2023almost}, we automatically get $6$-competitive universal algorithm for the prophet \emph{metric} facility location problem with a single sample per distribution. A similar result for the special case where all the distributions are identical is implied by previous work of \cite{DBLP:conf/soda/GargGLS08}.
    
    \section{Preliminaries}

    \label{sec:prelim}

    All logarithms in this paper are taken to be base $e$. In the following definitions, let $x,y \in \R^n_+$ be vectors. 
        The standard dot product between $x$ and $y$ is denoted
	$\langle x, y\rangle = \sum_{i=1}^n x_i y_i$. We use $\max(x,y)$ to denote the coordinate-wise maximum. We use a weighted generalization of KL divergence. Given
        a weight function $c$, define
	\[
        \wKL{x}{y} : = \sum_{i=1}^n c_i \left[x_i \log \left(\frac{x_i}{y_i}\right) -x_i + y_i\right].
    \]

    \paragraph{Augmentable Integer Programs.} A \emph{covering} integer program (\cip) is usually defined as an integer program (IP) for which the set of feasible solutions is upwards closed. We define the following more general class of problems which we call \emph{augmentable integer programs} (AIPs). These were studied in \cite{DBLP:conf/soda/GargGLS08}, but not given an explicit name.

        Let $V$ be a set of requests. For any subset of requests $V' \subseteq V$, let $\textsc{Sols}(V') \subseteq \Z^m$ be the subset of solutions that are feasible to $V'$. Next, for any subset of requests $V'$, any solution $z \in \textsc{Sols}(V')$, and any  request set $W$, define the augmentation cost
    \[
        \aug{W}{z, V'} = \min_w \left\{ c(w) \ \middle| \ \text{s.t. } \max(w,z) \in \textsc{Sols}(V' \cup W) \right\}, 
    \]
    or $\infty$ if no such $w$ exists.
    Let $\backup(W \mid z, V')$ be a minimizer when it exists.

    \begin{definition}[\aip] \label{defn:aip}
    An \emph{augmentable} integer linear program (\aip) is one in which augmentation costs are \emph{monotone}, i.e. for any $V' \subseteq V'' \subseteq V$, and any $z' \leq z''$ such that  $z' \in \textsc{Sols}(V')$ and $z'' \in \textsc{Sols}(V'')$, we have  $\aug{W}{z'', V''} \leq \aug{W}{z', V'}$ for any request set $W \subseteq V$.
    \end{definition}

    \begin{observation}[{\aip}s are subadditive]
    \label{obs:aip_subadd}
        For any $A, B \subseteq V$, we have $\opt(A \cup B) \leq \opt(A) + \opt(B)$.
    \end{observation}
    \begin{proof}
    We have that 
    \[\opt(A \cup B) \leq \opt(A) + \aug{B}{\opt(A),A} \leq \opt(A) + \opt(B).\]
    The first inequality follows since building a solution feasible to $A$ and then augmenting it to satisfy $B$ is only more expensive than $\opt(A \cup B)$. The second inequality follows from the monotonicity of augmentation costs property of {\aip}s, with $z' = \vec 0$, $z'' = \opt(A)$, $V' = \emptyset$, $V'' = A$, and $W = B$.
    \end{proof}

    Note that the standard IP formulation of \nmfl with indicator variables for $\{x_f\}_f$ for facilities, and $\{y_{fc}\}_{f,c}$ for facility-client connections is an \aip, but not a \cip. Likewise, \setmcov is an \aip but not a \cip.

    \paragraph{Online Models.}
    We briefly catalogue the various models that we treat in this paper.
    \begin{enumerate}[nosep]
        \item An \emph{online} \aip is an \aip in which some constraints are given upfront, and some are revealed sequentially over time. The algorithm must maintain a monotonically increasing solution that satisfies all the constraints revealed so far.
        \item A \emph{prophet} \aip instance is an online \aip instance in which the constraints $v^1, \ldots, v^n$ are drawn from known distributions $D^1, \ldots, D^n$.
        \item A \emph{k-sample} prophet \aip instance is an online \aip in which the constraints $v^1, \ldots, v^n$ are drawn from \emph{unknown} distributions $D^1, \ldots, D^n$, except the algorithm is given $k$ samples from each of the distributions before the online sequence begins.
        \item A \emph{free-order} (resp. $k$-sample) prophet \aip instance is a (resp. $k$-sample) prophet \aip instance in which the algorithm is allowed to adaptively decide the order in which it samples the known (resp. unknown but sampled $k$ times) distributions $D^1, \ldots, D^n$.
        \item A \emph{2-stage} prophet \aip instance is a prophet \aip instance and a positive number $\lambda$. The algorithm is allowed to purchase an initial solution $z_1$ before the online sequence begins, and any purchases $z_2$ made during the online sequence suffer a markup cost of $\lambda$.
        \item An \emph{online-with-a-sample} \aip instance is an online \aip instance in which, after the adversary fixes the input, the algorithm is given a uniformly random $\alpha$ fraction of the sequence upfront.
    \end{enumerate}
    \section{Universal Prophet Algorithms}

\label{sec:prophet}

In this section we prove theorem \cref{thm:main_prophet} via a reduction to random-order set cover. Recall that \cite{DBLP:conf/soda/AzarKW14} gave such a reduction for maximization problems. Our results are a complimentary attempt to do this for minimization problems.

In fact, our reduction is more powerful in two ways:
\begin{itemize}
    \item It holds for {\aip}s generally, beyond set cover. Hence to construct a prophet algorithm for an \aip, it suffices to construct a random-order algorithm.
    \item We require a weaker property even than random order. In fact, we can reduce the prophet setting to the \emph{free-order} prophet setting, where the algorithm is granted the freedom to choose the order in which it samples the distributions $D_1, \ldots, D_n$. Random-order algorithms are a special class of free-order algorithms.
\end{itemize}

\begin{theorem}
\label{thm:proph_to_customproph}
Let $\mathcal{I}$ be an instance class of prophet \aip. If algorithm $\mathcal{A}$ is a free-order prophet \aip algorithm that achieves competitive ratio $\Delta$ on class $\mathcal{I}$ using $k$ samples, then there is a fixed-order prophet \aip algorithm $\mathcal{A}'$ for class $\mathcal{I}$ achieving competitive ratio $2\Delta$ using $k+1$ samples.
\end{theorem}

In particular, random-order algorithms are $0$-sample free-order prophet algorithms, and furthermore \loc of \cite{DBLP:conf/focs/0001KL21} is a random-order set cover algorithm. Hence we get \cref{thm:main_prophet} as a corollary.

 \begin{proof}[Proof of \cref{thm:proph_to_customproph}]

 Let $\mathcal{A}$ be a $k$-sample algorithm for prophet \aip using custom/adaptive order $\pi$ and with expected competitive ratio $\Delta$. 
 Define $\mathcal{A}'$ to be \Cref{alg:proph_to_rosc}:

\makeatletter
\renewcommand{\ALG@name}{Reduction}
\makeatother

\begin{algorithm}[H]
\caption{\textsc{($k+1$)-sample prophet to $k$-sample free order prophet}}
\label{alg:proph_to_rosc}
\begin{algorithmic}[1]
    \State Train algorithm $\mathcal{A}$ on $k$ samples each of $D^1, \ldots, D^n$.
    \State Let $\textsc{MockRun} = \{\widehat v^1, \ldots, \widehat v^n\}$ be one sample each from $D^1, \ldots, D^n$. 
    \For{$\tau = 1, 2, \ldots, n$}
    \State $\pi^\tau \leftarrow$ the $\tau^{th}$ sample in (possibly adaptive) order $\pi$ specified by $\mathcal{A}$. \label{line:draw_freeorder}
    \State Feed $\pi^\tau$ to $\mathcal{A}$.
    \EndFor
    \State Let $\widehat z$ be the output of $\mathcal{A}$.
    \State Initialize $z \leftarrow \widehat z$.
    \For{$t=1,2,\ldots, n$}
    \State Draw $v^t \sim D^t$.  \label{line:draw_fixedorder}
    \If{$v^t$ not satisfied by $z$}
    \State Update $z \gets \max(z, \textsc{backup}(v^t \mid z, \textsc{MockRun}\cup\{v^1, \ldots, v^{t-1}\}))$. \label{line:proph_to_rosc_backup} 
    \EndIf
    \EndFor    
    \State \Return $z$.
\end{algorithmic}
\end{algorithm}

\makeatletter
\renewcommand{\ALG@name}{Algorithm}
\makeatother

Clearly $\mathcal{A}'$ uses $k+1$ samples, since that is enough samples to simulate $\mathcal{A}$: algorithm $\mathcal{A}$ requires $k$ samples upfront, and $\mathcal{A}'$ uses one more to simulate the real draw from each distribution. We turn to bounding the cost of $\mathcal{A}$.

The two sets of samples $v_1, \ldots, v_n$ and $\widehat v^1, \ldots, \widehat v^n$ are identically distributed, so \[\expect{\coptof{v_1, \ldots, v_n}} = \expect{\coptof{\widehat v^1, \ldots, \widehat v^n}}.\] Thus the expected cost of the solution $\widehat z$ bought by $\mathcal{A}$ is at most $\Delta \cdot \expect{\coptof{v_1, \ldots, v_n}}$, by the guarantee on $\mathcal{A}$. It remains to bound the cost of the backup purchases in \cref{line:proph_to_rosc_backup}.
To this end, consider each pair of requests $v^t, \widehat v^t \sim D^t$, where $v^t$ is drawn on \cref{line:draw_fixedorder}, and $\widehat v^t$ is part of the mock run specified on \cref{line:draw_freeorder}.
We will refer to these pairs of requests as \emph{mates}.

We will argue that the expected augmentation cost of a request $v^t$ is no more than the expected augmentation cost of its mate $\widehat v^t$ during the simulation of $\mathcal{A}$. Towards this, let $z(v)$ be the state of $z$ at the beginning of the round in which request $v$ arrives. Let $\widehat z(\widehat v)$ be the state of the solution of algorithm $\mathcal{A}$ at the beginning of round in which $\widehat v$ arrives in the simulation of $\mathcal{A}$ (that is, \emph{in the order $\pi$ chosen by $\mathcal{A}$}). Finally, let $\textsc{MockRun}_{<\widehat v}$ be the set of requests of $\textsc{MockRun}$ that arrive before $\widehat v$ according to order $\pi$. Now:
\begin{align*}
	\expect*{\aug{v^t}{z(v^t), \textsc{MockRun}\cup\{v^1, \ldots, v^{t-1}\}}}&\leq \expect*{\aug{v^t}{\widehat z(\widehat v^t),  \textsc{MockRun}_{<\widehat v^t}}} \\
 &= \expect*{\aug{\widehat v^t}{\widehat z(\widehat v^t), \textsc{MockRun}_{<\widehat v^t}}} 
	 \intertext{The inequality holds by the monotonicity of augmentation property of {\aip}s in \cref{defn:aip}, since $z(v^t) \geq \widehat z(\widehat v^t)$. The equality holds because $v^t$ and $\widehat v^t$ are identically distributed. Summing over $t$, and noting that $\mathcal{A}$ must pay at least $\aug{\widehat v^\tau}{\widehat z(\widehat v^\tau), \textsc{MockRun}_{<\widehat v^\tau}}$ in round $\tau$ in the event that $\widehat v^\tau$ is unsatisfied on arrival, we get that the total backup cost is bounded as}
	\sum_t \expect*{\aug{v^t}{z(v^t), \textsc{MockRun}\cup\{v^1, \ldots, v^{t-1}\}}} &\leq \sum_t \expect*{\{\aug{\widehat v^t }{\widehat z(\widehat v^t), \textsc{MockRun}_{<\widehat v^t}}} \\
 &\leq \expect{c(\widehat z)}, 
\end{align*}
which we can bound by $\Delta \cdot \expect{\coptof{v_1, \ldots, v_n}}$ by assumption of $\mathcal{A}$. Thus, in total, \cref{alg:proph_to_rosc} pays at most $2 \cdot \Delta \cdot \expect{\copt}$.
 \end{proof}

We note that this proof only requires the monotone augmentation cost property for individual requests.
    \section{Two-Stage Prophet Algorithms}

    \label{sec:2stage}

    Recall the 2-stage prophet setting. At the outset we have sample access to distributions $D_1, \ldots, D_t$, as well as some $\lambda > 0$.
    \begin{itemize}
        \item
        Stage 1: The algorithm may buy a partial solution $z_0$ and incur cost $c(z_0)$. 
        \item
        Stage 2: Requests $v_t \sim D_t$ arrive one-at-a-time and the algorithm must augment its solution to satisfy them immediately. If $z_1$ is the portion of the solution bought in this second phase, the algorithm incurs an additional cost of $\lambda \cdot c(z_1)$.
    \end{itemize}
    Note that for $0 < \lambda \leq 1$ the algorithm should always wait to buy sets online, and this reduces to the prophet setting above. 
    For $\lambda > 1$ we may assume without loss of generality that $\lambda$ is an integer (at the expense of a small constant factor).
    
    Our aim is to compete with $\opto$, the solution bought by the optimal online algorithm for this two-stage problem. We write  $\opto= \max(z_0^*, z_1^*)$ where the solution bought in advance $z_0^*$ is deterministic, and the solution bought during the online sequence $z_1^*$ depends on the realizations of the draws from the distributions. 

    Our main result in this section is that the $2$-stage setting is no harder than the random-order setting.
    
    \begin{theorem}
    \label{thm:2stage_to_rosc}
    Let $\mathcal{I}$ be an instance class of prophet \aip. If algorithm $\mathcal{A}$ is a random-order \aip algorithm that achieves competitive ratio $\Delta$ on class $\mathcal{I}$, then there is a $2$-stage prophet \aip algorithm $\mathcal{A}'$ for class $\mathcal{I}$ achieving competitive ratio $2\Delta$ with respect to the optimal online policy using $\lambda$ samples.
    \end{theorem}

    \begin{proof}
         Let $\mathcal{A}$ be the random-order algorithm for instance class $\mathcal{I}$ with expected competitive ratio $\Delta$. Define $\mathcal{A}'$ to be \cref{alg:two_stage_OSC}:

         \begin{algorithm}[H] \caption{\textsc{$2$-stage prophet to random order}}
		\label{alg:two_stage_OSC}
		\begin{algorithmic}[1]
            \Statex \textbf{First Stage:}
            \For{$i=1, \ldots, \lambda$}
            \State Let $\textsc{MockRun}_i \gets \{\widehat v_i^1, \ldots, \widehat v_i^n\}$ be one sample each from $D^1, \ldots, D^n$.
            \EndFor
            \State Let $\textsc{MockRun} = \bigcup_i \textsc{MockRun}_i$.
            \For{$\tau = 1, 2, \ldots, \lambda \cdot n$}
            \State $\pi^\tau \leftarrow$ the $\tau^{th}$ request of \textsc{MockRun} in random order $\pi$.
            \State Feed $\pi^\tau$ to $\mathcal{A}$.
            \EndFor
            \State Let $z_0$ be the output of $\mathcal{A}$. Buy $z_0$. \label{line:two-stage_LoC_call}
            \Statex{}
            \Statex \textbf{Second Stage:}
            \State Initialize $z_1 \leftarrow 0$.
                \For{$t=1,2,\ldots, n$}
                \State Draw $v^t \sim D^t$.  \label{line:2stage_draw_fixedorder}
                \State Let $z \gets \max(z_0, z_1)$.
                \If{$v^t$ not satisfied by $z$}
                \State Let $w \gets \textsc{backup}(v^t \mid \max(z_0,z_1), \textsc{MockRun} \cup \{v^1, \ldots, v^{t-1}\})$.
                \State Update $z_1 \gets \max(z_1, w)$.\label{line:2stage_backup} 
                \EndIf
                \EndFor    
                \State \Return $\max(z_0, z_1)$.
		\end{algorithmic}
	\end{algorithm}

        The proof proceeds in two steps. Let $Z := \expect{c(z^*_0) + \lambda \cdot c(z^*_1)}$ be the expected cost of $\opto$. First we bound the expected cost of $z_0$ computed by $\mathcal{A}$ in terms of $Z$, and then we bound the total cost of backups, i.e. $z_1$, in terms of the cost of $z_0$.

        For the first bound on $z_0$, we follow the ``boosted sampling'' argument of \cite{DBLP:conf/stoc/GuptaPRS04}.
        Suppose that $z_0^*$ is the optimal first-stage solution, and $z_1^*, \ldots, z_\lambda^*$ are the second-stage solutions bought by the optimal online strategy when fed each of the sequences $\textsc{MockRun}_i$. By the subadditivity property of {\aip}s from \cref{obs:aip_subadd},
        \begin{align*}
            \coptof{\textsc{MockRun}} &\leq c(z_0^*) + \sum_{i=1}^\lambda c(z_i^*) = c(z_0^*) + \expectover{i\sim [\lambda]}{\lambda\cdot  c(z_i^*))},
        \end{align*}
        where $i \sim [\lambda]$ denotes that $i$ is drawn uniformly from $[\lambda]$. Taking the expectation over the drawing of $\textsc{MockRun}$, we get that $\expect{\coptof{\textsc{MockRun}}} \leq Z$. Since algorithm $\mathcal{A}$ is $\Delta$-competitive, we immediately get that $\expect{c(z_0)} \leq \Delta \cdot Z.$

        It remains to bound the cost of $z_1$. This second half of the proof resembles that of \cref{thm:proph_to_customproph}. Let $z_0(v)$ denote the state of the solution held by $\mathcal{A}$ before the arrival of request $v$.

        Fix an index $t \in [n]$. Without loss of generality, reorder the corresponding $\textsc{MockRun}$ samples $\widehat v_1^t, \ldots, \widehat v_\lambda^t$ to agree with their relative order in $\pi$. Define $\textsc{MockRun}_{<v}$ to be the set of clients of $\textsc{MockRun}$ that arrive before $v$ according to order $\pi$. By \cref{defn:aip}, since $z_0(\widehat{v}_1^t) \leq \ldots \leq z_0(v^t)$ for every realization of the random variables, we have that for each request $v$,
        \begin{align*}
            \aug{v}{z_0(\widehat{v}_1^t), \textsc{MockRun}_{<\widehat{v}_1^t}} 
            \geq  \cdots &\geq \aug{v}{z_0(\widehat{v}_\lambda^t), \textsc{MockRun}_{<\widehat{v}_\lambda^t}} \\
            &\geq \aug{v}{z_0(v^t), \textsc{MockRun} \cup \{v^1, \ldots, v^{t-1}\}}. 
            \intertext{Then, taking the expectation over both the random sequences of $z_0(\widehat{v}_1^t),\ldots, z_0(v^t)$ and the identically distributed draws of $\widehat{v}^t_1, \ldots, \widehat{v}^t_\lambda, v^t \sim D^t$,}
            \expect*{\aug{\widehat{v}_1^t}{z_0(\widehat{v}_1^t), \textsc{MockRun}_{<\widehat{v}_1^t}}}\geq \cdots &\geq \expect*{\aug{\widehat{v}_\lambda^t}{z_0(\widehat{v}_\lambda^t), \textsc{MockRun}_{<\widehat{v}_\lambda^t}}} \\
            &\geq \expect*{\aug{v^t}{z_0(v^t), \textsc{MockRun} \cup \{v^1, \ldots, v^{t-1}\}}}.
        \end{align*}
        Summing yields
        \begin{align*}
        \lambda \cdot \expect*{\aug{v^t}{z_0(\widehat{v}^t), \textsc{MockRun} \cup \{v^1, \ldots, v^{t-1}\}}} &\leq \sum_{i=1}^\lambda \expect*{\aug{\widehat{v}_i^t}{z_0(\widehat{v}_i^t),\textsc{MockRun}_{<\widehat{v}_i^t}}}.
        \end{align*}
        Finally, summing again over $t \in [n]$, we get that the second-stage costs of \cref{alg:two_stage_OSC} are bounded by
        \begin{align*}\lambda \cdot \expect{c(z_1)} &= \sum_{t=1}^n \lambda \cdot \expect*{\aug{v^t}{z_0(v^t),\textsc{MockRun} \cup \{v^1, \ldots, v^{t-1}\}}} \\
        &\leq \sum_{t=1}^n\sum_{i=1}^\lambda \expect*{\aug{\widehat{v}_i^t}{z_0(\widehat{v}_i^t),\textsc{MockRun}_{<\widehat{v}_i^t}}} \\
        &\leq \expect{c(z_0)},\end{align*}
        where the last inequality holds since $\mathcal{A}$ pays at least $\aug{\widehat{v}_i^t}{z_0(\widehat{v}_i^t),\textsc{MockRun}_{<\widehat{v}_i^t}}$ in rounds where $\widehat{v}_i^t$ is unsatisfied on arrival.
        
        We conclude that the total expenditure of the algorithm is $\expect{c(z_0) + \lambda \cdot c(z_1)} \leq 2 \expect{c(z_0)} \leq 2 \Delta \cdot Z$.
    \end{proof}
    \section{Online-With-a-Sample}

\label{sec:was}

In this section we show a general reduction from online-with-a-sample {\aip}s to the random-order version.

\begin{theorem}
\label{thm:WaS_to_rosc}
Let $\mathcal{I}$ be an instance class of {\aip}s. If algorithm $\mathcal{A}$ is a random-order \aip algorithm with competitive ratio $\Delta$ on class $\mathcal{I}$, then there is an online-with-a-sample algorithm $\mathcal{A}'$ for class $\mathcal{I}$ with competitive ratio $\Delta/\alpha$.
\end{theorem} 

\begin{proof}[Proof of \cref{thm:WaS_to_rosc}]

Let $\mathcal{A}$ be the random-order algorithm for instance class $\mathcal{I}$ with expected competitive ratio $\Delta$. Define $\mathcal{A}'$ to be \cref{alg:was_to_rosc}:

\begin{algorithm}[H]
\caption{\textsc{Online-with-a-sample to Random Order}}
\label{alg:was_to_rosc}
\begin{algorithmic}[1]
    \State Let $\textsc{Samples} = \{\widehat v^1, \ldots, \widehat v^{\alpha \cdot n}\}$ be the samples given upfront. 
    \For{$\tau = 1, 2, \ldots, \alpha \cdot n$ }
    \State $\pi^\tau \leftarrow$ the $\tau^{th}$ sample in random order. \label{line:was_draw_ro}
    \State Feed $\pi^\tau$ to $\mathcal{A}$.
    \EndFor
    \State Let $\widehat z$ be the output of $\mathcal{A}$.
    \State Initialize $z \gets \widehat z$.
    \For{$t=1,2,\ldots, n$}
    \State Draw $v^t \sim D^t$.  \label{line:was_draw_fixedorder}
    \If{$v^t$ not satisfied by $z$}
    \State Update $z \gets \max(z, \textsc{backup}(v^t \mid z, \textsc{Samples} \cup \{v^1, \ldots, v^{t-1}\}))$. \label{line:was_to_rosc_backup} 
    \EndIf
    \EndFor    
    \State \Return $z$.
\end{algorithmic}
\end{algorithm}

Assume without loss of generality that $n$ is a multiple of $\nicefrac{1}{\alpha}$. We imagine generating $\textsc{Samples}$  according to the following procedure. First, perform a random partition of $v^1, \ldots, v^n$ into $\alpha \cdot n$ groups of size $\nicefrac{1}{\alpha}$. Pick a uniformly random representative request from each subset to include in $\textsc{Samples}$. For each request $v$, let $r(v)$ be the representative request of the group containing $v$, and let $\textsc{Samples}_{<v}$ be the set of samples that arrive before $v$ according to random order $\pi$. Now: 
\begin{align*}
	\expect*{\aug{v^t}{z(v^t),\textsc{Samples} \cup \{v^1, \ldots, v^{t-1}\}}}
 &\leq \expect*{\aug{v^t}{\widehat z(r(v^t)),\textsc{Samples}_{<r(v^t)}}} \\
 &= \expect*{\aug{r(v^t)}{\widehat z(r(v^t)),\textsc{Samples}_{<r(v^t)}}}.
 \intertext{The inequality holds by the monotonicity of augmentation condition of \cref{defn:aip}, since $z(v^t) \geq \widehat z(r(v^t))$. The equality holds because $v^t$ and $r(v^t)$ are identically distributed at the beginning of the run. Summing over $t$, and noting that $\mathcal{A}$ must pay at least $\aug{\widehat v^\tau}{\widehat z(v^{\tau})}$ in round $\tau$ when $\widehat v^\tau$ is unsatisfied on arrival, we get that the total backup cost is bounded as}
	\sum_{t=1}^n \expect*{\aug{v^t}{z(v^t),\textsc{Samples} \cup \{v^1, \ldots, v^{t-1}\}}} &\leq \sum_{t=1}^n \expect*{\aug{r(v^t)}{z(r(v^t)),\textsc{Samples}_{<r(v^t)}}} \\
    &= \sum_{\tau=1}^{\alpha\cdot n} \frac{1}{\alpha} \cdot \expect*{\aug{\widehat v^\tau}{z(\widehat v^\tau),\textsc{Samples}_{<\widehat v^\tau}}} \\
    &\leq \frac{1}{\alpha} \expect{c(\widehat z)}.
\end{align*}
The equality above comes from the fact that each representative that makes up $\textsc{Samples}$ appears $1/\alpha$ times in the first sum. We can bound the expression above by $\nicefrac{\Delta}{\alpha} \cdot \expect{\coptof{v_1, \ldots, v_n}}$ by the assumption on $\mathcal{A}$, and so in total, \cref{alg:proph_to_rosc} pays at most $\nicefrac{\Delta}{\alpha} \cdot \expect{\copt}$.
\end{proof}

\medskip

We conclude by remarking that a similar argument proves \cref{thm:WaS_to_rosc} in the slightly different `online-with-a-sample' setting in which each of the requests in the online sequence is sampled independently with probability $\alpha$.
More generally, for $\alpha \leq 1$ it is possible to get an analogous tradeoff between $\alpha$ in the number of samples and 
$\nicefrac{1}{\alpha}$ in the approximation ratio for \cref{thm:proph_to_customproph} and \cref{thm:2stage_to_rosc} as well.
    	\section{Online Facility Location}

	\label{sec:nmfl}

    In this section we apply our framework to online facility location problems.

    \subsection{Facility Location}
	In \textsc{FacilityLocation}, the input is a set of $n$ clients and $m$ facilities. Each facility $f$ has an opening cost $c_f$, and each client-facility pair $(v,f)$ has a connection cost $c_{fv}$. The goal is to open a number of facilities and connect each client to exactly one open facility such that the total cost is minimized. 
	A standard integer programming formulation for offline (unit-demand) \textsc{FacilityLocation} is as follows \cite{conforti2014integer}:
    \begin{align}
        \min \sum_f c_f \cdot x_f + &\sum_{f,v} c_{fv} \cdot y_{fv} \label{eq:nmfl_ip} \\
        \sum_j y_{fv} &\geq 1 \qquad \forall v \notag \\ 
        y_{fv} &\leq x_f \qquad \forall f,v \notag \\
        x_f, y_{fv} &\in \{0,1\}. \notag
    \end{align}

    In the random-order online version, the $m$ facilities are known ahead of time, and $n$ unknown clients arrive online in random order; on the arrival of each client, the algorithm must choose which (if any) new facilities to open, and then connect the client to an open facility. 
    Decisions are irrevocable, in the sense that a client may not change which facility it has connected to after arrival.

    We observe that \textsc{FacilityLocation} is amenable to our framework:
    \begin{restatable}{observation}{nmflisAIP}
		\label{lem:nmflisAIP}
		\textsc{FacilityLocation} is an \aip.
	\end{restatable}
    This enables us to convert algorithms for random-order \facloc into algorithms for the prophet, two-stage, and with-a-sample settings.

    \subsection{Metric Facility Location}
    
    Random-order \facloc is well studied when the connection costs $c_{fv}$ satisfy the triangle inequality. In pioneering work, Meyerson gave an $8$-approximation for this problem \cite{meyerson2001online}, which has recently been improved to a $3$-approximation \cite{kaplan2023almost}. 
    Appealing to \cref{lem:nmflisAIP} and plugging this algorithm as a black box into \cref{thm:proph_to_customproph,thm:2stage_to_rosc,thm:WaS_to_rosc}, we get:
    
    \begin{corollary}
        For the \mfl problem, there exists a $6$-competitive algorithm in the single-sample prophet setting, a $6$-competitive algorithm in the 2-stage prophet setting, and a $3/\alpha$-competitive algorithm in the online-with-a-sample setting.
    \end{corollary}

    In the next section, we study the more general problem without the metric assumption.
    
    \subsection{Non-Metric Facility Location}
    When connection costs do not satisfy the triangle inequality, this problem is more complex.
    In particular, \nmfl recovers \setcov as the special case in which all service costs are $c_{fv} \in \{0, \infty\}$. In this section, we give an algorithm for random-order \nmfl that is best possible, even in the special case of set cover.
    We show:

    \begin{theorem} \label{thm:nmfl_main_theorem}
        There exists an $O(\log mn)$-competitive algorithm for random-order \nmfl.
    \end{theorem}
    From \cref{lem:nmflisAIP} and \cref{thm:proph_to_customproph,thm:2stage_to_rosc,thm:WaS_to_rosc} we then directly obtain:
    \begin{corollary}
        For random-order \nmfl, there exists an $O(\log mn)$-competitive algorithm in the single-sample prophet setting, a $O(\log mn)$-competitive algorithm in the 2-stage prophet setting, and a $O(\log mn/\alpha)$-competitive algorithm in the online-with-a-sample setting.
    \end{corollary}

    \medskip

    We demonstrate the versatility of the the \loc algorithm of \cite{DBLP:conf/focs/0001KL21} for random-order set cover by adapting it to random-order \facloc, which is not a pure covering problem. The challenge is to decide which facilities to open in any given round; once this is decided, one can assume the incoming client in that round always connects to the cheapest open facility.

	Let $\mathcal{C}^\thist$ be the facilities purchased by the end of round $\thist$. For every client $v$, define 
	\begin{align*}
		f^t(v) &\in \argmin_{f} (\mathbbm{1}\{f \not \in \mathcal{C}^\thist\} \cdot c_f + c_{fv}) \\
		\kappa^\thist_v &:= \min_{f} (\mathbbm{1}\{f \not \in \mathcal{C}^\thist\} \cdot c_f + c_{fv})
	   \intertext{to be the facility which can connect the client in the cheapest way possible (including opening said facility if necessary), and the corresponding marginal cost of doing so. This is the cost at the end of round $\thist$. Note that $k_v^t$ corresponds to $\aug{v^t}{\mathcal{C}^t, \{v^1, \ldots, v^{t-1}\}}$ in our more general notation. Also define} 
		\Gamma^t(v) &:= \{f : c_{fv} \leq \kappa^t_v / 2\}
	\end{align*}
	to be the set of set of facilities that, if opened, would reduce the marginal cost of connecting $v$ by at least a factor of $2$. 
	Then we will say a facility $f$ \emph{covers} a client $v$ at time $t$ if $f \in \Gamma^t(v)$.

    We will show that \cref{alg:nmfl} is in expectation an $O(\log mn)$-approximation for random-order \nmfl. 
    Our approach may be viewed as running \loc algorithm for set cover, but on a \emph{dynamically changing} set system. The facilities are the sets, the clients are the elements, and a client's element is contained in a facility's set at time $t$ if $f \in \Gamma^t(v)$.
    
	\begin{algorithm}
	\caption{\textsc{LearnOrCoverNMFL}}
	\label{alg:nmfl}
	\begin{algorithmic}[1]
		\State Let $\mathcal{F}' \leftarrow \{f : \beta/m \leq c_f \leq \beta \}$ and let $m' \leftarrow |\mathcal{F}'|$.
		\State Initialize $x_f^{0} \leftarrow \frac{\beta}{c_f \cdot m'} \cdot \mathbbm{1}\{f \in \mathcal{F}' \}$.
		\For{$\thist=1,2\ldots, n$}
		\State{$v^\thist \leftarrow$ $\thist^{th}$ client in the random order, and let $\mathcal{R}^\thist \leftarrow \emptyset$.}
        \If{$\kappa_{v^t}^\lastt \geq \beta / t$} \label{line:nmfl_cheaply_connected_test}
        \State \ForEach facility $f$, add $\mathcal{R}^\thist \leftarrow \mathcal{R}^\thist \cup \{f\}$ with probability $\min(\kappa^\lastt_{v^\thist} \cdot x^{\lastt}_f/\beta,1)$. \label{line:nmfl_sample}
		\State Update $\mathcal{C}^\thist \leftarrow \mathcal{C}^{\lastt} \cup \mathcal{R}^\thist$.
		\If {$\sum_{f \in \Gamma^{\lastt}(v^{\thist})}  x^{\lastt}_f < 1$}
		\State For every facility $f$, update $x^{\thist}_f \leftarrow x^{\lastt}_f \cdot \exp\left\{\mathbbm{1}\{f \in \Gamma^{\lastt}(v^{\thist}) \} \cdot \kappa^{\lastt}_{v^{\thist}} / c_f\right\}$. 
		\State Let $Z^{\thist} = \langle c, x^\thist \rangle / \beta$ and normalize $x^{\thist} \leftarrow x^{\thist} / Z^{\thist}$. \label{line:nmfl_loc_cost_invariant}
		\Else
		\State $x^\thist \leftarrow x^{\lastt}$. \label{line:nmfl_noupdate}
		\EndIf
        \Else
		\State $x^\thist \leftarrow x^{\lastt}$. \label{line:nmfl_noupdate_cheap}
		\EndIf
          \State $\mathcal{C}^\thist \leftarrow \mathcal{C}^\thist \cup \{f^{\lastt}(v^{\thist})\}$ and connect $v^{\thist}$ to $f^\lastt(v^{\thist})$.
        		 \label{line:nmfl_backup}
		\EndFor 
        \State \Return $\mathcal{C}$.
	\end{algorithmic}
\end{algorithm}

	By a guess-and-double approach, we may assume the algorithm has access to a bound $\beta$ such that $\lpopt \leq \beta \leq 2 \cdot \lpopt$; here $\lpopt$ is the cost of the optimal solution for the linear programming relaxation of \eqref{eq:nmfl_ip} for the given \nmfl instance.
    We will denote by $\Xi^t$ the event that the arriving client $v^\thist$ satisfies the condition on \cref{line:nmfl_cheaply_connected_test} that $\kappa_{v^t}^\lastt \geq \beta/t$ and the bulk of \cref{alg:nmfl} is executed.
    We will say that $v^\thist$ is \emph{preemptively connected} if this condition is not met; this is the event $\neg \Xi^t$.

	Through \cref{line:nmfl_loc_cost_invariant} we maintain:
	\begin{invariant}
		\label{inv:nmfl_loc_cost}
		For all time steps $\thist$, it holds that $\langle c, x^\thist \rangle = \beta$.
	\end{invariant}

	We start by defining notation. 
    Let $\opt = (x^*, y^*)$ be an optimal fractional solution. 
    Let $U^{\thist} = \{v^{t+1}, \ldots, v^n\}$ be the clients remaining uncovered at the end of round $t$ (where $U^{0}=U$ is the entire client set). 
    Let $X^t(v) := \sum_{f\in \Gamma^{\lastt}(v)} x_f^{\lastt}$ be the fractional weight of facilities $f$ which cover $v$ at time $\lastt$.
    We define $\rho^t := \sum_v \kappa^t_v$, and consider the following potential:
    \begin{align}
    	\Phi(t) := C_1 \cdot \underbrace{\left(\wKL{x^*}{x^t} + 2 \cdot \sum_{v \in U^t} \sum_f c_{fv} \cdot y_{fv}^* \right)}_{\Phi_L(t)} + C_2 \cdot \underbrace{\beta \cdot \log \left(\frac{\rho^t}{\beta} + \frac{1}{n}\right)}_{\Phi_C(t)}, \notag
    \end{align}
    where the constants $C_1$ and $C_2$ will be determined later.
    We will refer to $\Phi_L$ as the ``learning'' portion of the potential and $\Phi_C$ as the ``covering'' portion of the potential. 

    Our potential $\Phi(t)$ resembles the one used to analyze \loc for set cover \cite{DBLP:conf/focs/0001KL21} which can also be decomposed into ``learning'' and ``covering'' portions. Their learning portion also involves a KL-divergence term, but ours is more intricate since we additionally charge to the connection cost paid by fractional \opt.

    \begin{restatable}[Bounds on $\Phi$]{lemma}{nmflphibounds}
		\label{lem:nmfl_phi_bounds}
		The initial potential is bounded as $\Phi(0) = O(\beta\cdot \log mn)$, and $\Phi(\thist) \geq -\beta \cdot \log n$ for all $\thist$.
	\end{restatable}

	We now show the potential decreases sufficiently in every round. 
	We bound the decrease of each term in the potential separately.

	\begin{restatable}[Change in $\Phi_L$]{lemma}{nmflphil}
		\label{lem:nmfl_klchange}
		For rounds when the event $\Xi^t$ holds, the expected change in $\Phi_L$ is
		\begin{align}
			&\expectover{v^\thist, \mathcal{R}^{\thist}}*{ \Phi_L(\thist) - \Phi_L(\lastt) \mid x^{\lastt}, U^{\lastt}, \Xi^t} \notag \\
			&\leq \expectover{v \sim U^{\lastt}}*{\frac{e^2-1}{2} \cdot \kappa^{\lastt}_v \cdot  \min\left(X^t(v), \: 1\right) - \kappa^{\lastt}_v}. \label{eq:nmfl_learning_goal}
            \end{align}
            When the arriving $v^\thist$ is preemptively connected and $\Xi^t$ does not hold,
            \begin{align}
            \Phi_L(\thist) - \Phi_L(\lastt) \label{eq:nmfl_learning_cheapchange}
            &\leq 0.
			\end{align}
	\end{restatable}
	Note that the expected change in the statement above depends only on the randomness of the arriving uncovered client $v^\thist$, not on the randomly chosen facilities $\mathcal{R}^{\thist}$. On the other hand we can bound the change in $\Phi_C$ as follows.
	
    \begin{restatable}
    [Change in $\Phi_C$]{lemma}{nmflphic}
		\label{lem:nmfl_weighted_utchange}
		For all rounds $\thist$ for which $\Xi^t$ holds, the expected change in $\Phi_C$ is
		\begin{align}
            \expectover{v^{\thist}, \mathcal{R}^{\thist}}*{
            \Phi_C(\thist) - \Phi_C(\lastt)
            \mid x^{\lastt}, U^{\lastt}, \Xi^t} 
            &\leq - \frac{1-e^{-1}}{4}\cdot \expectover{u \sim U^{\lastt}}*{\kappa^{\lastt}_u \cdot \min\left(X^\thist(v), \: 1\right)}. \label{eq:nmfl_covering_goal}
            \end{align}
            For rounds in which $\Xi^t$ does not hold,
            \begin{align}
            \Phi_C(\thist) - \Phi_C(\lastt) \leq 0. \label{eq:nmfl_covering_cheapchange}
        \end{align}
	\end{restatable}

 We defer the proof of \cref{lem:nmfl_phi_bounds,lem:nmfl_klchange,lem:nmfl_weighted_utchange} to \cref{sec:deferred}, and we now show how to combine them to prove the theorem. 
    \begin{proof}[Proof of \cref{thm:nmfl_main_theorem}]
    Let $c(\alg(\thist))$ be the cost paid by \cref{alg:nmfl} up to and including time $\thist$, and furthermore let $c_{\textsc{pre}}(\alg(\thist))$ denote the cost paid by \cref{alg:nmfl} for clients which are preemptively connected on \cref{line:nmfl_backup} up to and including round $\thist$, and $c_{\textsc{LoC}}(\alg(\thist))$ denote the cost paid by \cref{alg:nmfl} on \cref{line:nmfl_sample,line:nmfl_backup} up to and including round $\thist$ during rounds in which $v^t$ is not preemptively connected and $\Xi^t$ holds. (\textsc{Pre} is for preemptive, \textsc{LoC} is for \textsc{LearnOrCover}).
    
    We can provide a simple bound on the cheap facility-client connections which the algorithm buys on \cref{line:nmfl_backup} in the event that $v^t$ is preemptively connected (that is, in the event that $\Xi^t$ does not hold).
    These connections collectively cost at most 
    \begin{align}        
        c_{\textsc{pre}}(\alg(n)) &= c_{\textsc{pre}}(\alg(n)) - c_{\textsc{pre}}(\alg(0)) \notag \\
        &=\sum_{\thist} \left(c_{\textsc{pre}}(\alg(\thist)) - c_{\textsc{pre}}(\alg(\lastt))\right) \notag \\
        &\leq \sum_{t \in [n]} \frac{\beta}{t}
        = O(\beta\cdot \log n). \label{eq:nmfl_cheap_costs_bound}
    \end{align}
	We now consider the per-round costs incurred by the bulk of the algorithm, during the rounds in which $\Xi^t$ holds.
    In every round $\thist$, the expected cost of the sampled facilities $\mathcal{R}^\thist$ in \cref{line:nmfl_sample} is ${\kappa^\lastt_{v^\thist} \cdot \langle c, x^{\lastt}\rangle / \beta = \kappa^\lastt_{v^\thist}}$ (by \cref{inv:nmfl_loc_cost}). 
    The algorithm pays at most an additional $\kappa^\lastt_{v^\thist}$ in \cref{line:nmfl_backup}, and hence the total expected cost per round is at most $2\cdot \kappa^{\lastt}_{v^\thist}$.
	
	By combining \cref{lem:nmfl_klchange,lem:nmfl_weighted_utchange}, and setting the constants $C_1 = 2$ and $C_2 = 4e(e+1)$, we have
	\begin{align}
		&\expectover{\substack{v^{\thist}, \mathcal{R}^{\thist}}}*{\Phi(\thist) - \Phi(\lastt) | v^{1}, \ldots, v^{\lastt}, \mathcal{R}^{1}, \ldots, \mathcal{R}^{\lastt}, \Xi^\thist} \notag \\
		& = \expectover{\substack{v^{\thist}, \mathcal{R}^{\thist}}}*{
			\begin{array}{ll}
				&C_1\cdot \left(\Phi_L(\thist) - \Phi_L(\lastt)\right) \\
				+ &C_2\cdot \left(\Phi_C(\thist) - \phi_C(\lastt)\right)
			\end{array}
			| v^{1}, \ldots, v^{\lastt}, \mathcal{R}^{1}, \ldots, \mathcal{R}^{\lastt}, \Xi^\thist} \notag \\
		&\leq - \expectover{\substack{v^{\thist}, \mathcal{R}^{\thist}}}*{2 \cdot \kappa^{\lastt}_{v^{\thist}} | v^{1}, \ldots, v^{\lastt}, \mathcal{R}^{1}, \ldots, \mathcal{R}^{\lastt}, \Xi^\thist}, \notag
	\end{align}
	which cancels the expected change in $c_{\textsc{LoC}}$ in each round. 
    We therefore have the inequality
	\begin{equation}
        \expectover{\substack{v^{\thist}, \mathcal{R}^{t}}}*{\Phi(\thist) - \Phi(\lastt) + c_{\textsc{LoC}}(\alg(\thist)) - c_{\textsc{LoC}}(\alg(\lastt)) | v^{1}, \ldots, v^{\lastt}, \mathcal{R}^{1}, \ldots, \mathcal{R}^{\lastt}} \leq 0, \label{eq:nmfl_potential_plus_cost_supermartingale}
    \end{equation}
    where we used that the change in $\Phi_L$ and $\Phi_C$ is at most $0$ for rounds in which $\Xi^\thist$ does not hold.
    
    By repeatedly applying \eqref{eq:nmfl_potential_plus_cost_supermartingale} for all $1 \leq \thist \leq n$, we obtain
    \begin{align}
        \expectover{v, \mathcal{R}}{\Phi(n) - \Phi(0) + c_{\textsc{LoC}}(\alg(n)) - c_{\textsc{LoC}}(\alg(0))} &\leq 0 \notag \\
        \expectover{v, \mathcal{R}}{c_{\textsc{LoC}}(\alg(n))} &\leq \Phi(0) + c_{\textsc{LoC}}(\alg(0)) - \expectover{v, \mathcal{R}}{\Phi(n)} \notag \\
        &\leq O(\beta \cdot \log mn) + 0 - \left(- \beta \cdot \log n\right), \label{eq:nmfl_expected_costbound}
    \end{align}
    where \eqref{eq:nmfl_expected_costbound} follows from the fact that $c_{\textsc{LoC}}(\alg(0)) = 0$, together with the bounds on $\Phi$ established in \cref{lem:nmfl_phi_bounds}.

    To conclude, by 
    \eqref{eq:nmfl_cheap_costs_bound} and \eqref{eq:nmfl_expected_costbound} we have
    \begin{align}
        \expectover{v, \mathcal{R}}{c(\alg(n))} &= 
        \expectover{v, \mathcal{R}}{c_{\textsc{pre}}(\alg(n))} + \expectover{v, \mathcal{R}}{c_{\textsc{LoC}}(\alg(n))} \notag \\
        &\leq \beta + O(\beta \cdot \log n) + O(\beta \cdot \log mn),
    \end{align}
    as desired.
\end{proof}

This concludes our discussion of random-order \facloc. Our results settle the approximability of both the metric and non-metric versions this problem in the random-order model, and hence also in the prophet, 2-stage prophet, and online-with-a-sample models.
    
	\section{Set Multicover with Multiplicity Constraints}
	\label{sec:setmcov}

    We now give a second application of our framework to set multicover with multiplicity constraints, which we refer to as \setmcov. 
    This generalizes unit-cost set cover and additionally introduces non-covering constraints; in particular, multiplicity constraints on the decision variables.
    The set multicover problem can be formally written as the following IP:
	\begin{align}
		\begin{array}{lll}
			\min_z &\langle 1, z \rangle		 \label{eq:CIPunitbox}\\
			\text{s.t.} &A z \geq b \\
			&z \in \{0, 1\}^m,
		\end{array}
	\end{align}
    where the entries of $A$ are in $\{0,1\}$. It is then without loss of generality to assume that $b \in \Z_+^n$. 
    We recover \setcov as the special case where $b = 1$.
    We again observe that \textsc{SetMultiCover} is amenable to our framework:
    \begin{restatable}{observation}{SMCisAIP} \label{lem:smcisAIP}
		\setmcov is an \aip.
	\end{restatable}
    This enables us to convert algorithms for random-order \setmcov into algorithms for the prophet, two-stage, and with-a-sample settings.

	In random-order \setmcov the rows of $A$ are revealed in random-order, and the algorithm must maintain a monotonically increasing solution $z$ that satisfies all constraints revealed thus far. We show:

    \begin{theorem} \label{thm:boxcip_main_theorem}
        There exists an $O(\log mn)$-competitive algorithm for random-order \setmcov.
    \end{theorem}
    By appealing to \cref{lem:smcisAIP} and  \cref{thm:proph_to_customproph,thm:2stage_to_rosc,thm:WaS_to_rosc} we then directly obtain:
    \begin{corollary}
        For \setmcov, there exists an $O(\log mn)$-competitive algorithm in the single-sample prophet setting, a $O(\log mn)$-competitive algorithm in the 2-stage prophet setting, and a $O(\log(mn)/\alpha)$-competitive algorithm in the online-with-a-sample setting.
    \end{corollary}
    
    \medskip
    Once again we show how to adapt ideas from \cite{DBLP:conf/focs/0001KL21} to this more general setting. This time we extend their more general algorithm for random-order {\cip}s.
    We will make use of the preliminaries given in \cref{sec:prelim}, and our approach will be similar to that of \cref{sec:nmfl}.
    
	Let $z^t$ denote the integer solution in round $t$. 
    Let $d_i^t$ denote the undercoverage of $i$ in the beginning of round $t$; that is $d_i^t := \max(0, b_i - \langle a_i, z^{t-1} \rangle) $ where $z^t$ is the integer solution at the end of round $t$. Now we define
	\[
		\rho^t := \sum_{i \in [n]} d_i^t.
	\]
	
	We pursue a guess-and-double approach to identifying $\beta$ such that $\lpopt \leq \beta \leq 2\cdot\lpopt$, where $\lpopt$ is the cost of an optimal fractional solution $x^*$ to \eqref{eq:CIPunitbox}. We will maintain a solution $x^t$ to \eqref{eq:CIPunitbox} of cost $\beta$. 
	\begin{algorithm}[H]
	\caption{\textsc{LearnOrCoverSMC}}
	\label{alg:unitcostcip}
	\begin{algorithmic}[1]
		\State Initialize $x_j^0 \leftarrow \frac{\beta}{m}$ for every $j$, and set $z_j^0 \leftarrow 0$.
		\For{$t=1,2,\ldots$}		
		\State{$i \leftarrow$ $t$-th constraint in the random order.}
		\If {$i$ is not covered on arrival} \label{line:setmcov_covered_on_arrival}
			\State Let $\mathcal{T}^t :=\{j : z_j^{t-1} = 0, a_{ij} = 1\}$. \Comment{unbought coordinates covering $i$}
			\State Let $X_i^{t-1} := \sum_{j \in \mathcal{T}^t} x_j^{t-1}$. \Comment{frac. coverage by unbought $j$}
            \State $d_i^t := b_i - \langle a_i, z^{t-1}\rangle$. \Comment{integral uncoverage}
			\State \ForEach $j$, sample  $z_j^t \leftarrow \Ber(d_i^\thist \cdot x_j^{t-1} /\beta)$. 
            \label{line:unitcip_randsample}
			\If {$X_i^{t-1} \leq d_i^t$ } \Comment{$i$ fractionally undercovered} \label{line:unitcip_noupdate}
				\State Set $z_j^t \leftarrow 1$	for $j \in \mathcal{T}^t$ with  $x_j^{t-1} \geq 1/e$. \label{line:boxcip_bigfracsets}
				\State \ForEach $j$, update $x_j^t \leftarrow x_j^{t-1} \cdot \exp\left\{ \mathbbm{1}\{j \in \mathcal{T}^t\} \right\}$.
				\State Let $Z^{(t)} := \langle 1, x^t \rangle / \beta$ and renormalize $x^t \leftarrow x^t/Z^t$. 
            \EndIf
			\If{$i$ still uncovered}
				\State $x^\thist \leftarrow x^\lastt$ and $z^\thist \leftarrow z^\lastt$.
                \State $z_j^t \leftarrow 1$ for $d_i^t$-many arbitrary $j\in \mathcal{T}^t$. \Comment{buy backup sets} \label{line:boxcip_backup}
			\EndIf
		\Else
            \State $x^\thist \leftarrow x^\lastt$ and $z^\thist \leftarrow z^\lastt$.
        \EndIf
		\EndFor 
        \State \Return $z$
	\end{algorithmic}
\end{algorithm}
    
    In a manner similar to \cref{alg:nmfl}, our algorithm will react differently to constraints which arrive uncovered. 
    In \cref{alg:nmfl} the criterion was that arriving clients cannot be cheaply connectable; here we perform a \loc step if the element arrives undercovered, meaning that $d_{i^\thist}^\thist > 0$ on \cref{line:setmcov_covered_on_arrival}.
    
	\begin{theorem}
	    \label{thm:unit_cost_cip}
		For set multicover with multiplicity constraints, \cref{alg:unitcostcip} achieves an expected competitive ratio of $O(\log mn)$.
	\end{theorem}
	
	Our potential is 
	\begin{align}
    	\Phi(t) := C_1 \cdot \underbrace{\KL{x^*}{x^t}}_{\Phi_L(t)} + C_2 \cdot \underbrace{\beta \cdot \log \left(\frac{\rho^t}{\beta} + \frac{1}{m}\right)}_{\Phi_C(t)}, \notag
    \end{align}
	where again we view $\Phi_L$ is the learning part of the potential and $\Phi_C$ is the covering part.
    We will fix constants $C_1$ and $C_2$ later.
	To begin, we bound the value of this potential:
	\begin{restatable}[Bounds on $\Phi$] {lemma}{smcovphibounds} \label{lem:unitcip_phi_bounds}
		The initial potential is bounded as $\Phi(0) = O(\beta\cdot \log mn)$, and $\Phi(t) \geq -\beta\cdot \log(n)$ for all rounds $\thist$.
	\end{restatable}

Let $U^t$ denote the constraints $i$ which have not yet arrived in round $t$; that is, $U^t = \{i^\thist, i^{\thist + 1}, \ldots, i^n\}$. 
We now turn to the expected change in the learning portion of the potential in each round.

\begin{restatable}[Change in $\Phi_L$]{lemma}{mcovphil}
	\label{lem:unit_cost_setmcov_klchange}
	For rounds in which $i^t$ arrives uncovered, the expected change in $\Phi_L$ is
	\begin{align}
        \expectover{\substack{i^t, \mathcal{R}^{t}}}*{ \Phi_L(\thist) - \Phi_L(\lastt) \mid x^{t-1}, U^{t-1}, d_{i^\thist}^\thist > 0}
        &\leq \expectover{i \sim U^{t-1}}{(e-1) \cdot \min(X^{t-1}_i, \: d_i^t) - d_i^t}. \label{eq:setmcov_learning_inequality}
        &
        \intertext{When $i^t$ is covered on arrival,}
        \Phi_L(\thist) - \Phi_L(\lastt) &\leq 0. \label{eq:setmcov_learning_nostep}
	\end{align}
\end{restatable}	
	
We next bound the expected change in $\Phi_C(t)$. 

\begin{restatable}[Change in $\Phi_C$]{lemma}{mcovphic}
	\label{lem:unit_cost_setmcov_utchange}
	In every round for which $i^\thist$ is uncovered on arrival, the expected change in $\Phi_C$ is
	\begin{align}
        \expectover{\substack{i^t,\mathcal{R}^{t}}}*{\Phi_C(t) - \Phi_C(t-1) \mid x^{t-1}, U^{t-1}, d_{i^\thist}^\thist > 0} 
        &\leq - \frac{\gamma}{2} \cdot \expectover{i \sim U^{t-1}}*{\min\left(X_{i}^{t-1}, \: d_i^t \right)}, \label{eq:setmcov_covering_inequality} \\
        \intertext{where $\gamma$ is a fixed constant. On rounds in which $i^\thist$ arrives covered,}
        \Phi_C(t) - \Phi_C(t-1) &\leq 0. \label{eq:setmcov_covering_nostep}
	\end{align}
\end{restatable}

    We again defer the proofs of \cref{lem:unitcip_phi_bounds,lem:unit_cost_setmcov_klchange,lem:unit_cost_setmcov_utchange} to \cref{sec:deferred}, and now show how to combine them to bound the expected cost of \cref{alg:unitcostcip}:
    
	\begin{proof}[Proof of \cref{thm:unit_cost_cip}]
    Combining \cref{lem:unit_cost_setmcov_klchange} and \cref{lem:unit_cost_setmcov_utchange}, choosing $C_1 = (e + 2)$ and $C_2 = 2\cdot(e+2)(e-1) / \gamma$, and recalling that $d_{\thist^t}^\thist = 0$ in rounds for which $\thist^t$ arrives covered, we have that
	\begin{align}
		&\expectover{\substack{i^t, \mathcal{R}^{t}}}*{\Phi(t) - \Phi(t-1) | i^{1}, \ldots, i^{t-1}, \mathcal{R}^{1}, \ldots, \mathcal{R}^{t-1}} \notag \\
		& = \expectover{\substack{i^t, \mathcal{R}^{t}}}*{
		\begin{array}{ll}
			&C_1 \cdot \left(\Phi_L(\thist) - \Phi_L(\lastt)\right) \\
			+ & C_2 \cdot \left(\Phi_C(\thist) - \Phi_C(\lastt)\right)
		\end{array}
		| i^{1}, \ldots, i^{t-1}, \mathcal{R}^{1}, \ldots, \mathcal{R}^{t-1}} \notag \\
		&\leq - \expectover{\substack{i^t, \mathcal{R}^{t}}}*{(e+2) \cdot d_i^t | i^{1}, \ldots, i^{t-1}, \mathcal{R}^{1}, \ldots, \mathcal{R}^{t-1}} \label{eq:setmcov_expected_potential_change}
    \end{align}
    for all rounds $\thist$.
    In each round the algorithm buys at most $e \cdot d_i^t$ coordinates in \cref{line:boxcip_bigfracsets} (since this only happens in the case when $X_i^{t-1} \leq d_i^\thist$), and samples $d_i^\thist$ sets in expectation in \cref{line:unitcip_randsample}, and buys at most $d_i^t$ coordinates in \cref{line:boxcip_backup}, for a total of at most $(e+2)\cdot d_i^t$ sets bought in expectation.
    From \eqref{eq:setmcov_expected_potential_change} we therefore have that for all rounds $\thist$,
	\begin{align}
		\expectover{\substack{i^\thist \sim U^{(t)} \\ \mathcal{R} \sim x^{(t)}}}*{\Phi(t) - \Phi(t-1) + c(\alg(t)) - c(\alg(t-1)) | i^{(1)}, \ldots, i^{(\lastt)}, \mathcal{R}^{(1)}, \ldots, \mathcal{R}^{(\lastt)}} \leq 0. \label{eq:setmcov_potential_cost_supermartingale}
	\end{align}
	Repeatedly applying \eqref{eq:setmcov_potential_cost_supermartingale} for all $1 \leq t \leq n$ yields
    \begin{align}
        \expectover{i, \mathcal{R}}*{\Phi(n) - \Phi(0) + c(\alg(n)) - c(\alg(0))} &\leq 0 \notag \\
        \expectover{i, \mathcal{R}}*{c(\alg(n))} &\leq c(\alg(0)) + \Phi(0) - \expectover{i, \mathcal{R}}*{\Phi(n)}, \notag
        \intertext{and so observing that $c(\alg(0)) = 0$ and applying \cref{lem:unitcip_phi_bounds} we have}
        \expectover{i, \mathcal{R}}*{c(\alg(n))} &\leq O(\beta \cdot \log mn),
    \end{align}
    as desired.
\end{proof}

This concludes our discussion of \setmcov. Once again our results settle the approximability of this problem in the random-order, prophet, 2-stage prophet, and online-with-a-sample models. This gives a partial answer towards the question of \cite{DBLP:conf/focs/0001KL21} on whether the same results are possible for general box-constrained {\cip}s; it remains a tantalizing open question to understand the general case.
    \section{Conclusion}

In this paper we showed that stochastic set cover can be solved ``even more obliviously'' than \cite{DBLP:journals/siamcomp/GrandoniGLMSS13}, with only coarse advice about the process generating the input. It is tempting to try to relax these online-with-advice models for set cover by allowing for bounded error in the advice. We discuss why this is challenging in \cref{sec:lowerbounds}, as some natural candidates for relaxed models have strong lower bounds.

We submit as an interesting open problem the task of determining the tight dependence on $\alpha$ in \cref{thm:main_was}. We conjecture that it should be $O(\log(mn)\log(\nicefrac{1}{\alpha}))$. \cref{thm:main_was} implies that when $\alpha = \Theta(1)$ there is an $O(\log mn)$ competitive algorithm, and when $\alpha = \nicefrac{1}{\poly(n)}$ the $O(\log m \log n)$ competitive algorithm of \cite{DBLP:journals/mor/BuchbinderN09} is best possible; this conjecture interpolates smoothly between these extremes.

         \textbf{Acknowledgements} Roie Levin would like to thank Guy Even for asking about the 2-stage prophet model, and Niv Buchbinder for helpful discussions.

        \appendix
        \section{Deferred Proofs}
\label{sec:deferred}

Here we present the proofs of lemmas supporting our random-order algorithms of \cref{sec:nmfl,sec:setmcov}.

\subsection{Facility Location}

\nmflphibounds* 

\begin{proof}
        To begin, we claim that for all $f$ in the support of $x^*$, we have $c_f \leq \beta$.
        To see this, consider an optimal fractional solution $(x^*, y^*)$ to \eqref{eq:nmfl_ip}, and assume for the sake of contradiction that there is some facility $\hat{f} \in \supp{x^*}$ for which $x_f^* > \beta$.
        This $\hat{f}$ provides some fractional connection to some clients; let $\epsilon := \min( y_{\hat{f},v}^* :y_{\hat{f},v}^* > 0)$ be the minimum connection provided to all such clients.
        Finally, consider the perturbed solution given by setting $x_{\hat{f}}' = x_{\hat{f}}^* - \epsilon$, setting $x_{f}' = \min(x_f^* + \epsilon, 1)$ for $f \neq \hat{f}$, and setting $y_{fv}' = \min(y_{f,v}^* + \epsilon, 1)$ for all $f,v$. 
        This solution $(x', y')$ remains feasible for \eqref{eq:nmfl_ip} and costs strictly less than $\beta$; therefore no such $\hat{f}$ exists.

        Having established that $c_f \leq \beta$ for all $f$, we know that $\supp{x^*} \subseteq \supp{x^0}$. 
        We then bound the initial $KL$-divergence term by
        \begin{align*}
            \wKL{x^*}{x^0} 
            &= \sum_f c_f x_f^* \log \left( x_f^* \frac{c_f m'}{\beta}\right)  + \sum_f c_f (x_f^0 - x_f^*) \\
            &\leq \beta (\log m + 1),
        \end{align*}
        where we used that $x_f^* c_f \leq \beta$ above, that $\langle c, x^*\rangle = \lpopt \leq \beta$, and that $\langle c, x^0\rangle = \beta$.

        Next we consider the second term in $\phi_L$.
        This is the second half of the objective of the fractional relaxation \eqref{eq:nmfl_ip}, evaluated at the optimal solution; therefore it is bounded above by $\lpopt \leq \beta$.
        
        We turn to the last term, $\beta \cdot \log \left(\sum_v \kappa^0_v
        /\beta + 1/n \right)$.
        In order to show that this is at most $O(\beta \cdot \log n)$, it suffices to demonstrate that $\kappa_v^0 \leq \beta$ for all clients $v$.
        This can be seen by considering the relaxation of \eqref{eq:nmfl_ip} to serving only the client $v$, which is the problem of fractionally finding the cheapest augmentation for $v$ at time $0$. 
        Since this is a relaxation of \eqref{eq:nmfl_ip}, its optimal solution will cost at most $\beta$.
        Finally, this relaxation integrally chooses the facility $f_v := \arg\min_f c_f + c_{fv}$; therefore $\kappa_v^0 \leq \beta$.
        (Since augmentation costs only decrease, indeed $\kappa_v^t \leq \kappa_v^0 \leq \beta$ for all $v$ and $t$.)

        \medskip
        We conclude with the lower bound on $\Phi$.
        Both terms in $\Phi_L$ are nonnegative, and so $\Phi_C \geq - \beta \cdot \log n$ since $\rho^\thist \geq 0$.
\end{proof}

\nmflphil*

\begin{proof} 
        Inequality \eqref{eq:nmfl_learning_cheapchange} is straightforward: when \cref{line:nmfl_cheaply_connected_test} does not execute, there is no change to $x^\thist$, and so the KL term is unchanged. 
        At the same time, the fractional optimum term only decreases.
 
        Our main task is to prove \eqref{eq:nmfl_learning_goal}.
        We break the proof into cases. 
        Let $\Lambda^t$ be the event that $X^\thist(v^\thist) < 1$. If $\Lambda^t$ does not hold, in \cref{line:nmfl_noupdate} we set the vector $x^{\thist} = x^{\lastt}$, so the change in KL term is again unchanged. 
        This means that inequality \eqref{eq:nmfl_learning_goal} holds trivially, since $(e^2 - 1)/2 > 1$. 
        Henceforth we focus on the case when $\Lambda^t$ holds.
		
		Recall that the expected change in relative entropy depends only on the arriving uncovered element $v^{\thist}$.
        Beginning with the KL term and expanding definitions, and writing $v = v^{\thist}$ when it is clear from context,
		\begin{align}
			& \expectover{v^{\thist}, \mathcal{R}^{\thist}}*{\wKL{x^*}{x^{\thist}} - \wKL{x^*}{x^{\lastt}} \mid x^{\lastt}, U^{\lastt}, \Lambda^t} \notag \\
			&= \expectover{v \sim U^{\lastt}}*{\sum_{f} c_f \cdot x^*_f \cdot \log \frac{x^{\lastt}_f}{x_f^{\thist}} |\Xi^t, \Lambda^t} \notag  \\
			&= \expectover{v \sim U^{\lastt}}*{\langle c, x^*\rangle \cdot  \log Z^{\thist} - \sum_{f \in \Gamma^{\lastt}(v)} c_f \cdot x^*_f \cdot \log e^{\kappa^{\lastt}_v/c_f} |\Xi^t, \Lambda^t} \notag  \\
			&\leq \expectover{v \sim U^{\lastt}}*{
				\begin{array}{ll}
					\displaystyle \beta \cdot \log\left(\sum_{f \in \Gamma^{\lastt}(v)} \frac{c_f}{\beta} \cdot x_f^{\lastt}  \cdot e^{\kappa^{\lastt}_v/c_f} + \sum_{f \not \in \Gamma^{\lastt}(v)} \frac{c_f}{\beta} \cdot x_f^{\lastt} \right) \\
					- \displaystyle \sum_{f \in \Gamma^{\lastt}(v)} \kappa^{\lastt}_v  \cdot x^*_f
				\end{array}
				  |\Xi^t, \Lambda^t }\label{eq:nmfl_wkl_copt}, \\
			\intertext{where in the last step \eqref{eq:nmfl_wkl_copt} we expanded the definition of $Z^{\thist}$, and used $\langle c, x^*\rangle \leq \beta$.
            Then we can further bound \eqref{eq:nmfl_wkl_copt} by}
			&\leq  \expectover{v \sim U^{\lastt}}*{ \beta  \cdot \log\left(\sum_f \frac{c_f}{\beta} \cdot  x_f^{\lastt} + \frac{e^2-1}{2} \cdot \frac{\kappa^{\lastt}_v}{\beta} \cdot X^\thist(v) \right) - \displaystyle \sum_{f \in \Gamma^{\lastt}(v)} \kappa^{\lastt}_v  \cdot x^*_f
            |\Xi^t, \Lambda^t} \label{eq:nmfl_wkl_eapx}, \\
			\intertext{where we use the approximation $e^a \leq 1+(e^2-1) \cdot a/2$ for $a \in [0,2]$ (note that $\kappa^{\lastt}_v$ is the cheapest marginal connection cost for $v$, so for any $f \in \Gamma^{\lastt}(v)$, meaning that $c_{fv} \leq \kappa^{\lastt}_v/2$, we have that $c_f \geq \kappa^{\lastt}_v / 2$ and thus $\kappa^{\lastt}_v / c_f \leq 2$). Finally, using \cref{inv:nmfl_loc_cost}, along with the approximation $\log(1+y) \leq y$, we bound \eqref{eq:nmfl_wkl_eapx} by}
			&\leq \expectover{v \sim U^{\lastt}}*{\frac{e^2-1}{2} \cdot  \kappa^{\lastt}_v\cdot X^\thist(v) - \displaystyle \sum_{f \in \Gamma^{\lastt}(v)} \kappa^{\lastt}_v  \cdot x^*_f
            |\Xi^t, \Lambda^t} \notag  \\
			&\leq \expectover{v \sim U^{\lastt}}*{\frac{e^2-1}{2} \cdot \kappa^{\lastt}_v \cdot \min\left(X^\thist(v), \: 1\right) - \displaystyle \sum_{f \in \Gamma^{\lastt}(v)} \kappa^{\lastt}_v  \cdot x^*_f
            |\Xi^t, \Lambda^t}, \label{eq:nmfl_wkl_xvl1}
		\end{align}
		where \eqref{eq:nmfl_wkl_xvl1} follows by the definition of the event $\Lambda^t$. 

        \medskip
        We now turn to the second part of $\Phi_L$.
        The change to fractional optimum term in each round is $- 2 \cdot \sum_f c_{f v^t} \cdot y_{f v^t}^*$.
        Combining this with \eqref{eq:nmfl_wkl_xvl1} gives
        \begin{align}
            &\expectover{v \sim U^{\lastt}}*{\Phi_L(\thist) -\Phi_L(\lastt) |\Xi^t, \Lambda^t} \notag\\
            &\leq \expectover{v \sim U^{\lastt}}*{\kappa^{\lastt}_v \cdot \frac{e^2-1}{2} \cdot \min\left(X^\thist(v), \: 1\right) - \left(\sum_{f \in \Gamma^{\lastt}(v)} \kappa^{\lastt}_v  \cdot x^*_f + 2 \cdot \sum_f c_{f v^t} \cdot y_{f v}^* \right)
            |\Xi^t, \Lambda^\thist}.\notag
            \intertext{Since the fractional connection $v$ receives from outside of $\Gamma^{\lastt}(v)$ is at cost at least $\kappa_v^{\lastt}$ by definition, we may bound this by}
            &\leq \expectover{v \sim U^{\lastt}}*{\kappa^{\lastt}_v \cdot \frac{e^2-1}{2} \cdot \min\left(X^\thist(v), \: 1\right) - \left(\sum_{f \in \Gamma^{\lastt}(v)} \kappa^{\lastt}_v  \cdot x^*_f + 2 \cdot \sum_{f \not \in \Gamma^{\lastt}(v)} \frac{\kappa_v^{\lastt}}{2} \cdot y_{f v}^* \right)
            |\Xi^t, \Lambda^\thist} \notag \\
            &\leq \expectover{v \sim U^{\lastt}}*{\kappa^{\lastt}_v \cdot \frac{e^2-1}{2} \cdot \min\left(X^\thist(v), \: 1\right) - \kappa^{\lastt}_v 
            |\Xi^t, \Lambda^\thist},\label{eq:nmfl_learn_final}
		\end{align}
        where \eqref{eq:nmfl_learn_final} follows because the fractional connection $v$ receives from outside of $\Gamma^{\lastt}(v)$ in $(x^*, y^*)$ is at least $1- \sum_{f \in \Gamma^{\lastt}(v)} x_f^*$.

        We have shown the lemma statement both when $\Lambda^t$ holds and when it does not, which completes the proof.
\end{proof}

\nmflphic*

\begin{proof}	
        \Cref{eq:nmfl_covering_cheapchange} is once again straightforward, since $\rho^\thist$ is monotonically decreasing in $\thist$.
        We therefore focus on proving \eqref{eq:nmfl_covering_goal}.
        
        We start by considering the expected change to $\Phi_C$ over the randomness of the sampling, for a fixed arriving client $v$.
        Expanding definitions,
		\begin{align}
			& \expectover{\mathcal{R}^{\thist}}*{\Phi_C(\thist) - \Phi_C(\lastt) \mid x^{\lastt}, U^{\lastt}, \Xi^t, v^{\thist} = v} \notag \\
			&= \beta \cdot \expectover{\mathcal{R}^{\thist}}*{\log \left(1-\frac{\rho^{\lastt} - \rho^{\thist}}{\rho^{\lastt} + \frac{\beta}{n}} \right) | U^{\lastt}, \Xi^t, v^{\thist} = v} \notag \\
			&\leq - \beta \cdot \frac{1}{\rho^{\lastt} + \frac{\beta}{n}} \cdot \expectover{\mathcal{R}^{\thist}}*{\rho^{\lastt} - \rho^{\thist} | U^{\lastt}, \Xi^t, v^{\thist} = v}. \label{eq:nmfl_ut_logapx} \\
			\intertext{Above, \eqref{eq:nmfl_ut_logapx} follows from the approximation $\log (1-y) \leq -y$. Expanding the definition of $\rho^{\thist}$, \eqref{eq:nmfl_ut_logapx} is bounded by }
			&\leq - \frac{\beta}{\rho^{\lastt} + \frac{\beta}{n}} \cdot \expectover{\mathcal{R}^{\thist}}*{\sum_{u \in U^{\lastt}} \frac{\kappa^{\lastt}_u}{2} \cdot \mathbbm{1}\{\Gamma^{\lastt}(u) \cap \mathcal{R}^\thist \neq \emptyset\}  | U^{\lastt}, \Xi^t, v^{\thist} = v} \notag \\
			&= - \frac{\beta}{\rho^{\lastt} + \frac{\beta}{n}} \sum_{u \in U^{\lastt}} \frac{\kappa^{\lastt}_u}{2} \cdot \probover{\mathcal{R}^{\thist}}{\Gamma^{\lastt}(u) \cap \mathcal{R}^\thist \neq \emptyset \mid U^{t-1}, \Xi^t, v^{\thist} = v} \notag \\
			&\leq - \frac{\beta}{\rho^{\lastt} + \frac{\beta}{n}} \sum_{u \in U^{\lastt}} \frac{\kappa^{\lastt}_u }{2}  \cdot (1-e^{-1}) \cdot  \min\left(\frac{\kappa^{\lastt}_v}{\beta} \cdot \sum_{f \in \Gamma^{\lastt}(u)}  x^{\lastt}_f, \: 1\right) \label{eq:nmfl_ut_prob}\\
			&\leq -  \frac{1-e^{-1}}{2}\cdot \kappa^{\lastt}_v \cdot \frac{|U^{\lastt}|}{\rho^{\lastt} + \frac{\beta}{n}} \cdot \expectover{u \sim U^{\lastt}}*{\kappa^{\lastt}_u \cdot \min\left(X^\thist(u), \: 1\right)} \label{eq:nmfl_ut_simplify} .
		\end{align}
		Step \eqref{eq:nmfl_ut_prob} is due to the fact that each facility $f$ is sampled independently with probability $\min(\kappa^{\lastt}_v x^{\lastt}_f /\beta, \: 1)$, so the probability any given client $u \in U^{t-1}$ gets at least one facility from $\Gamma^{\lastt}(u)$ is
		\begin{align*}
		1 - \prod_{f \in \Gamma^{\lastt}(u)} \left(1 - \min\left(\frac{\kappa^{\lastt}_v x^{\lastt}_f}{\beta},\:1\right)\right) &\geq 1 - \exp\left\{-\min\left(\frac{\kappa^{\lastt}_v}{\beta} \sum_{f \in \Gamma^{\lastt}(u)} x^{t-1}_f,\:1\right)\right\} \\
		& \stackrel{(**)}{\geq} (1 - e^{-1}) \cdot \min\left(\frac{\kappa^{\lastt}_v}{\beta} \sum_{f \in \Gamma^{\lastt}(u)} x^{t-1}_f,\:1\right).
		\end{align*}
		Above, $(**)$ follows from convexity of the exponential. 
		Step \eqref{eq:nmfl_ut_simplify} then follows by rewriting the sum as an expectation and using the fact that $\kappa^{\lastt}_v / \beta \leq 1$, which is justified (and used) in the proof of \cref{lem:nmfl_phi_bounds}.

		Taking the expectation of \eqref{eq:nmfl_ut_simplify} over $v^{\thist} \sim U^{\lastt}$, and using the fact that $\expectover{v \sim U^{\lastt}}*{\kappa^{\lastt}_v} = \rho^{\lastt} / |U^{\lastt}|$, the expected change in $\Phi_C$ becomes
		\begin{align}
			&\expectover{v^{\thist}, \mathcal{R}^{\thist}}*{\Phi_C(\thist) - \Phi_C(\lastt) \mid x^{\lastt}, U^{\lastt}, \Xi^t} \notag\\
			&\leq -  \frac{1-e^{-1}}{2} \cdot \frac{\rho^{\lastt}}{\rho^{\lastt} + \frac{\beta}{n}} \cdot \expectover{u \sim U^{\lastt}}*{\kappa^{\lastt}_u \cdot \min\left(X^\thist(u), \: 1\right)}, \notag\\ 
			&\leq -  \frac{1-e^{-1}}{4} \cdot \expectover{u \sim U^{\lastt}}*{\kappa^{\lastt}_u \cdot \min\left(X^\thist(u), \: 1\right)}, \label{eq:nmfl_covering_final_step}
		\end{align}    
		Where in \eqref{eq:nmfl_covering_final_step} we finally use the fact that the event $\Xi^t$ holds.
        Since this is the case, we know that $v^\thist$ is not preemptively connected in round $\thist$, and so $\rho^\lastt \geq \kappa_{v^\thist}^\lastt \geq \beta/t \geq \beta/n$, and so $\frac{\rho^\lastt}{\rho^\lastt + \frac{\beta}{n}} \geq \frac{1}{2}$.
        This is the claimed bound.
\end{proof}

\subsection{Set Multicover}

\smcovphibounds*

\begin{proof}
		We start with the upper bound, and address each portion of the potential in turn.
        First, $\KL{x^*}{x^0} = \sum_j x_j^* \log\left( \frac{x_j^*}{x_j^0} \right) = \sum_j x_j^* \log\left( m \frac{x_j^*}{\beta} \right) \leq \sum_j x_j^* \log m \leq \beta \log m$, since $x_j^* \leq \lpopt \leq \beta$ for all $j$. 
		Second, $\beta \cdot \log\left( \frac{\rho^0}{\beta} + \frac{1}{m}\right) \leq \beta \cdot \log\left( \sum_i \frac{b_i}{\beta} \right) \leq \beta \log n$, since $b_i \leq \lpopt \leq \beta$ for all $b_i$.

        We now justify the lower bound. 
        The KL divergence is nonnegative, as is $\rho^\thist$; therefore $\Phi(t) \geq \beta\cdot \log(1/m) = -\beta \cdot \log m$.
	\end{proof}

\mcovphil*

\begin{proof} 
    We first show \eqref{eq:setmcov_learning_nostep}. This holds because if $i^t$ is covered on arrival then 
    
    We break the proof of \eqref{eq:setmcov_learning_inequality} into two cases. 
    If $X_i^{t-1} \geq d_i^t$, by \cref{line:unitcip_noupdate}, the vector $x^{t}$ is not updated in round $t$, so the change in KL divergence is 0 which means that 
	\begin{align} &\expectover{i^t, \mathcal{R}^t}*{ \KL{x^*}{x^{t}} - \KL{x^*}{x^{t-1}} \mid x^{t-1}, U^{t-1}, X_{i^t}^{t-1}  \geq d_{i^t}^t} \notag \\
		&\leq \expectover{i \sim U^{t-1}}{(e-1) \cdot \min(X_i^{t-1},d_i^t) - d_i^t | X_i^{t-1} \geq  d_i^t} ,\label{eq:kl_xiged}\end{align}
	implying \eqref{eq:setmcov_learning_inequality} trivially. 
    Henceforth we focus on the case $X_{i}^{t-1} < d_i^t$.
	
	Recall that the expected change in relative entropy depends only on the arriving uncovered element $i^t$. Expanding definitions,
	\begin{align}
		& \expectover{i^t, \mathcal{R}^t}*{\KL{x^*}{x^{t+1}} - \KL{x^*}{x^{t}} \mid x^{t-1}, U^{t-1}, X_{i^t}^{t-1} < d_i^t} 
		\notag \\
		&= \expectover{i \sim U^{t-1}}*{\sum_j x^*_j \cdot \log \frac{x^{t-1}_j}{x^{t}_j} | X_i^{t-1} < d_i^t} 
		\notag  \\
		&= \expectover{i \sim U^{t-1}}*{\sum_j x^*_j \cdot  \log Z^{t} - \sum_{j \in \mathcal{T}^t} x^*_j \cdot \log e | X_i^{t-1} < d_i^t} 
		\notag  \\
		&\leq \expectover{i \sim U^{t-1}}*{\beta \cdot  \log Z^{t} - \sum_{j \in \mathcal{T}^t} x^*_j | X_i^{t-1} < d_i^t} 
		\label{eq:new1}  \\
		&\leq \expectover{i \sim U^{t-1}}*{\beta \cdot  \log Z^{t} - d_i^t | X_i^{t-1} < d_i^t}
		\label{eq:new2}, \\
		\intertext{where in \eqref{eq:new1} we used $\langle 1, x^*\rangle \leq \beta$. Since $x^*$ is a feasible fractional set cover, we know that $\sum_{j} a_{ij} x^*_j \geq b_i$, and since $x_j^* \leq 1$ this implies that $\sum_{j \in \mathcal{T}^t} x_j^* \geq d_i^t$, giving \eqref{eq:new2}. Expanding $Z^t$, we have that}
		&= \expectover{i \sim U^{t-1}}*{\beta \cdot  \log \left( \frac{1}{\beta} \sum_j x_j^{t-1} + \frac{(e-1)}{\beta} \sum_{j \in \mathcal{T}^t} x_j^{t-1} \right) - d_i^t | X_i^{t-1} < d_i^t}
		\notag	\\
		&= \expectover{i \sim U^{t-1}}*{\beta \cdot  \log \left( 1 + \frac{(e-1)}{\beta} X_i^{t-1} \right) - d_i^t | X_i^{t-1} < d_i^t}
		\label{eq:kl_eapx_cip} \\
		\intertext{Finally the approximation $\log(1+y) \leq y$ allows us to bound \eqref{eq:kl_eapx_cip} by}
		&\leq  \expectover{i \sim U^{t-1}}*{(e-1) \cdot X_i^{t-1} - d_i^t | X_i^{t-1} < d_i^t} \notag 
		\\
		&=  \expectover{i \sim U^{t-1}}{(e-1) \cdot \min(X_i^{t-1},d_i^t) - d_i^t | X_i^{t-1} < d_i^t}. \label{eq:kl_xild}
	\end{align}
	The lemma statement follows by combining \eqref{eq:kl_xiged} and \eqref{eq:kl_xild} using the law of total expectation.
\end{proof}

\mcovphic*

We will require an additional fact, the proof of which appears in \cite[Appendix A]{DBLP:conf/focs/0001KL21}:
\begin{fact}
	\label{lem:ip_expectcov}
	Given probabilities $p_j$ and coefficients $b_j\in [0,1]$, let
	$W := \sum_j b_j \Ber(p_j)$ be the sum of independent weighted
	Bernoulli random variables.
	Let $\Delta \geq (e-1)^{-1}$ be some constant.
	Then
	\[
	\expect*{\min\left(W, \Delta \right)} \geq \gamma \cdot \min\left(\expect*{W}, \Delta \right),
	\]
	for a fixed constant $\gamma$ independent of the $p_j$ and $b_j$.
\end{fact}

\begin{proof}[Proof of \cref{lem:unit_cost_setmcov_utchange}]
	When $i^t$ is covered on arrival \eqref{eq:setmcov_covering_nostep} holds trivially, since in this case $d_{i^\thist}^\thist = 0$ and so no sets $j$ are bought and $\rho^\thist = \rho^\lastt$.
    We therefore focus on proving \eqref{eq:setmcov_covering_inequality} in the case when $i^t$ is uncovered on arrival.
    
    Conditioned on $i^t = i$, the expected change in $\log \rho^{t}$ depends only on $\mathcal{R}^t$:
	\begin{align}
		& \expectover{i^t, \mathcal{R}^t}*{\Phi_C(t) - \Phi_C(t-1) \mid x^{t-1}, U^{t-1}, i^{t} = i} \notag \\
		&= \beta \cdot \expectover{\mathcal{R}^t}*{\log \left(1-\frac{\rho^{t-1} - \rho^{t}}{\rho^{t-1} + \frac{\beta}{m}} \right) | U^{t-1}, i^t = i} \notag \\
			&\leq - \beta \cdot \frac{1}{\rho^{t-1} + \frac{\beta}{m}} \expectover{\mathcal{R}^{t}}*{\rho^{t-1} - \rho^{t} | U^{t-1}, i^t = i}
			\label{eq:box_ip_ut_logapx}. \\
			\intertext{Above, \eqref{eq:box_ip_ut_logapx} follows from the approximation $\log (1-y) \leq -y$. Expanding definitions again, we further bound \eqref{eq:box_ip_ut_logapx} by}
			&= - \beta \cdot \frac{1}{\rho^{t-1} + \frac{\beta}{m}} \expectover{\mathcal{R}^t}*{\sum_{i' \in U^{t-1}} (d_{i'}^{t-1} - d_{i'}^t)} \label{eq:box_ip_ut_rhodef} \\
			&= - \beta \cdot \frac{1}{\rho^{t-1} + \frac{\beta}{m}} \sum_{i' \in U^{t-1}} \expectover{\mathcal{R}^t}{d_{i'}^{t-1} - d_{i'}^t} 
			\notag \\
			&\leq - \beta \cdot \frac{1}{\rho^{t-1} + \frac{\beta}{m}} \sum_{i' \in U^{t-1}} \expectation_{\mathcal{R}^t}\left[ \min\left(\sum_{j \in \mathcal{T}^t} \Ber(x_j^{t-1} \cdot d_i^{t-1}/\beta), \: d_{i'}^t\right)\right].
			\label{eq:box_ip_sampledef} \\
			\intertext{Here \eqref{eq:box_ip_sampledef} follows from the preceding line by the definition of the random sampling performed in \cref{line:unitcip_randsample} (with inequality because the algorithm buys other coordinates also).
			This Bernoulli probability is well-defined because $d_{i'}^{t-1} \leq \lpopt \leq \beta$, and since \cref{line:boxcip_bigfracsets} guarantees that $x_j^{t-1} \leq 1$ for all $j$ in $\mathcal{T}^t$. 
			The expectation of this sum of Bernoullis is $\sum_{j \in \mathcal{T}^t} x_j^{t-1} d_i^{t-1}/\beta = X_i^{t-1} d_i^{t-1}/\beta$.
			Since $d_i^{t-1} \geq 1$, we may apply \Cref{lem:ip_expectcov} to obtain}
			&\leq - \frac{1}{\rho^{t-1} + \frac{\beta}{m}} \sum_{i' \in U^{t-1}} \gamma \cdot  \min\left(\frac{d_i^{t-1}}{\beta} \cdot X_{i'}^{t-1}, \: d_{i'}^{t-1} \right) \label{eq:box_ip_randcovg} \\
			&\leq - \gamma \cdot d_i^{t-1} \cdot \frac{1}{\rho^{t-1} + \frac{\beta}{m}} \sum_{i' \in U^{t-1}} \min\left(X_{i'}^{t-1}, d_{i'}^{t-1}\right) \notag \\
			&= - \gamma \cdot d_i^{t-1} \cdot \frac{|U^{t-1}|}{\rho^{t-1} + \frac{\beta}{m}} \cdot \expectover{i' \sim U^{t-1}}*{\min \left(X_{i'}^{t-1}, d_{i'}^{t-1})\right)} \notag \\
            &\leq - \gamma \cdot d_i^{t-1} \cdot \frac{|U^{t-1}|}{2 \cdot \rho^{t-1}} \cdot \expectover{i' \sim U^{t-1}}*{\min \left(X_{i'}^{t-1}, d_{i'}^{t-1})\right)}, \label{eq:box_ip_kcrs}
		\end{align}
		where \eqref{eq:box_ip_kcrs} follows from the observation that $i^t$ arrives uncovered, so $\rho^{\lastt} \geq d_i^\thist \geq 1$ while $\beta/m \leq 1$; therefore $\frac{1}{\rho^{\lastt} + \beta/m} \leq \frac{1}{2\cdot \rho^{\lastt}}$.
		Taking the expectation of \eqref{eq:box_ip_kcrs} over $i \sim U^{t-1}$, and using the fact that $\expectover{i \sim U^{t-1}}*{d_i^{t-1}} = \rho^{t-1} / |U^{t-1}|$, the expected change in $\log \rho^t$ becomes
		\begin{align*}
			&\expectover{i^t, \mathcal{R}^t}*{\Phi_C(t) - \Phi_C(t-1) \mid x^{t-1}, U^{t-1}} 
			\leq - \frac{\gamma}{2} \cdot \expectover{i \sim U^{t-1}}*{\min\left(X_{i}^{t-1}, d_i^t \right)},
		\end{align*}    
		as desired.
	\end{proof}	
 	\section{Lower Bounds for Relaxed Models}

	\label{sec:lowerbounds}

    \subsection{Adversarial Corruptions}

	It is tempting to try to extend the with-a-sample  model to the case where the samples are noisy. In this section, we study one natural model, and show that, sadly, no randomized algorithm can achieve a competitive ratio of $o(\log m \log n)$ in polynomial time, unless $\P = \NP$.

	We begin by describing the model, which we call \nscws. The adversary begins by committing to an online set cover input sequence of length $n$ from a set system with $N$ elements and $m$ sets. A uniformly random $\alpha\cdot n$ of the input $S_1$ is then sampled. The adversary then chooses a set $S_2$ of size $\delta \cdot n$. The algorithm is given the set of constraints $S = S_1 \cup S_2$ in advance. Finally, the online sequence begins.

    We now construct our hard instance for \nscws when $\alpha, \delta = \Omega(1)$. We will use as a sub-instance the construction from \cite{korman2004use} which shows an $\Omega(\log m \log n)$ lower bound for the original online set cover problem. Let $(\mathcal{U}_{hard}, \mathcal{S}_{hard})$ be the underlying set system in the \cite{korman2004use} instance. Let $\sigma_{hard}$ be the sequence of elements given to the algorithm, and let the unordered set of these elements be called $U_{hard} \subseteq \mathcal{U}_{hard}$. In the construction of \cite{korman2004use}, $U_{hard}$ is a random variable such that $|U_{hard}| = \Theta(\sqrt{|\mathcal{U}_{hard}|})$.
    
    Our construction is the following. The set system consists of a set $S_0$ containing $N-\delta\sqrt{N}$ elements, together with a copy of the set system $(\mathcal{U}_{hard}, \mathcal{S}_{hard})$ with parameters $n' := |\mathcal{U}_{hard}| = \delta \sqrt{N}$ and $m' := |\mathcal{S}_{hard}| = \poly(n')$. The adversary commits to the online sequence which reveals all the elements of $S_0$ in arbitrary order, and then the elements of $U_{hard} \subseteq \mathcal{U}_{hard}$ in order $\sigma_{hard}$. This sequence is of length $\Theta(N)$. The adversary picks $S_2$, which is of size $\delta \cdot N \geq \delta \sqrt{N}$, to be all the elements of $\mathcal{U}_{hard}$. Hence no matter what the realization of the sample $S_1$ is, the algorithm has no information about the identity of $U_{hard}$ and must cover the hard online set cover sequence $\sigma_{hard}$ of length $\Theta(\sqrt{n'}) = \Theta(\delta^{1/2} N^{1/4})$ with no useful advice. Since any polynomial-time randomized online set cover algorithm has competitive ratio $\Omega(\log n' \log m')$ on $\sigma_{hard}$, no algorithm can achieve competitive ratio $o(\log n \log m)$ for \nscws when $\alpha, \delta = \Omega(1)$.

    \subsection{Relaxed Random Order}

    Another interesting question is whether \cite{DBLP:conf/focs/0001KL21} can be made to work when the input ordering is not fully uniformly random, but only nearly so.

    We show that the entropy of the arrival order distribution is \emph{not} a good parametrization of the distance to random order, in that there exist instances and distributions over arrival orders with nearly full entropy, $(1-\eps) n \log n \approx \log(n!)$, but for which any online algorithm has competitive ratio $\Omega(\log (\eps m) \log (\eps n))$. 
    
    One simple such instance is the following. There are $(1-\eps) n$ dummy elements presented in uniformly random order, followed by a hard online set cover sequence $\sigma_{hard}$ of length $\eps n$. The permutation distribution has the desired near maximal entropy, but no randomized polynomial time algorithm has  competitive ratio $o(\log (\eps m) \log (\eps n))$ unless $\P = \NP$, by the lower bound of \cite{korman2004use}.
            \section{Error in \cite{dehghani2018greedy}}
    \label{sec:error}

    In \cite{dehghani2018greedy} the authors claim an $O(\log n)$-competitive algorithm for prophet set cover via a reduction from prophet set cover to known i.i.d. set cover, which is the special case when all distributions are identical. 
    (They also claim an $O(1)$-competitive algorithm for prophet \mfl via the same reduction.)
    Their proof of this reduction, which appears in \cite[Theorem 9.10]{ehsani2017online} relies on the following claim.

    Let $D^1, \ldots, D^n$ be a sequence of distributions over elements of $\mathcal{U}$, and let $D^* = \frac{1}{n} \sum_{i=1}^n D^i$ be the average distribution. Let $\opt_{pht}$ be the expected size of the optimal set cover for $U$ when $U$ is formed by drawing one element from each of $D^1, \ldots, D^n$. Let $\opt_{iid}$ be the size of the optimal set cover when $U$ is formed by drawing $n$ times from $D^*$. The claim is that $\opt_{iid} \leq \opt_{pht}$, and we now show that this does not hold in general.
    
    Consider the instance with universe $(i,j)$ for $i \in [2]$ and $j \in [2]$ and the set system
    \[
        \mathcal{S} = 
        \{\{(1,1), (2,1)\}, 
        \{(1,1), (2,2)\}, 
        \{(1,2), (2,1)\},
        \{(1,2), (2,2)\}\}.
    \]
    Let the prophet distributions be $D^1$ and $D^2$, uniform distributions over $(1,1), (1,2)$ and $(2,1), (2,2)$ respectively. 
    Let $D^* = \frac{1}{2} (D^1 + D^2)$ be the average distribution. Then
    \begin{align*}
        \opt_{pht} &= 1, \qquad\qquad
        \expect{\opt_{iid}} = \frac{5}{4}.
\intertext{   More generally, let the universe be $(i,j)$ for $i \in [n]$ and $j \in [\log n]$ and the collection of sets be all sequences $\mathcal{S} = [\log n]^{[n]}$. 
    Then the $n$ prophet distributions $D^i$ are each uniform distributions over $\{(i,1), (i,2), \ldots, (i,\log n)\}$, and the average distribution is $D^* = \frac{1}{n}\sum_i D^i$ as before. 
    A balls-and-bins argument shows that}
        \opt_{pht} &= 1, \qquad\qquad
        \expect{\opt_{iid}} = \Omega\left(\frac{\log n}{\log\log n}\right),
    \end{align*}
    which demonstrates that the claimed inequality is in the worst case violated by at least this multiplicative factor.
 
	{\footnotesize
		\bibliography{dblp,refs}
		\bibliographystyle{alpha}
	}
\end{document}